%% file: main.tex
\ifdefined\pdfminorversion
\fi

\documentclass[acmsmall,screen,nonacm]{acmart}

\setcopyright{none}
\acmDOI{}
\acmISBN{}
\authorsaddresses{}

\makeatletter
\if@ACM@newfonts
\else
  \PackageError{HornCraft}{Required ACM font packages are missing}{%
    Install libertine, inconsolata (zi4), and newtx before compiling.}
\fi
\makeatother
\usepackage{mathtools}
\usepackage{booktabs}
\newcommand{\evalCount}[1]{\makebox[3.5em][r]{#1}}
\newcommand{\evalOutcomes}[2]{\makebox[1em][r]{#1} S + #2 U}
\usepackage{placeins}
\usepackage{comment}
\usepackage{xparse}
\usepackage{listings}
\usepackage{xcolor}
\usepackage{tikz}
\usetikzlibrary{arrows.meta,calc,positioning}

\floatstyle{ruled}
\newfloat{algorithm}{tbp}{loa}
\floatname{algorithm}{Algorithm}
\floatstyle{plain}

\newcounter{algline}[algorithm]

\newcommand{\algline}[1]{\refstepcounter{algline}\label{#1}\thealgline}
\newcommand{\alglineref}[1]{\hyperref[#1]{line~\ref*{#1}}}
\newcommand{\alglinesref}[2]{\hyperref[#1]{lines~\ref*{#1}--\ref*{#2}}}
\usepackage[colorinlistoftodos]{todonotes}

\apptocmd{\thebibliography}{\clubpenalty=10000\widowpenalty=10000}{}{}

\IfFileExists{mathpartir.sty}{
  \usepackage{mathpartir}
}{
  \def\bidpurulelabel right=##1\relax{##1}
  \NewDocumentCommand{\inferrule}{s O{} m m}{%
    \frac{\begin{gathered}##3\end{gathered}}{\begin{gathered}##4\end{gathered}}
    \;\text{\scriptsize\textsc{\expandafter\bidpurulelabel##2\relax}}%
  }
}

\graphicspath{{figures/}}

\title{LLM-Assisted Bi-Approximate CHC Solving}

\author{Baiting Jiang}
\affiliation{
  \institution{Tsinghua University}
  \city{Beijing}
  \country{China},
  \href{mailto:jbt24@mails.tsinghua.edu.cn}{jbt24@mails.tsinghua.edu.cn}
}
\email{jbt24@mails.tsinghua.edu.cn}

\begin{document}
\raggedbottom

\input{sections/0-Abstract}



\maketitle

\input{sections/1-Introduction}
\input{sections/2-Overview}
\input{sections/3-Preliminaries}
\input{sections/4-Compressing-CHC-Derivations}

\input{sections/4.1-Derivation-Power-Under-approximations}

\input{sections/4.2-Witness-Carrying-DPUs-for-Linear-CHCs}
\input{sections/4.3-Compositional-Construction-of-DPUs}
\input{sections/4.4-DPUs-for-Nonlinear-CHCs-and-Reach-Fact-Reuse}

\input{sections/5-LLM-Assisted-Bi-Approximate-CHC-Solving}

\input{sections/5.1-LLM-Guided-Reachability-Search}
\input{sections/5.2-LLM-Guided-Invariant-Synthesis}

\input{sections/5.3-The-HornCraft-Algorithm}
\input{sections/6-Evaluation}
\input{sections/6.1-Overall-Effectiveness}
\input{sections/6.2-Scalability-on-Deep-and-Branching-Derivations}
\input{sections/6.3-Ablation-of-DPU}
\input{sections/6.4-Limitations}
\input{sections/7-Related-Work}
\input{sections/8-Conclusion}

\bibliographystyle{ACM-Reference-Format}
\bibliography{references}

\end{document}

%% file: sections/0-Abstract.tex
\begin{abstract}
Constrained Horn clause (CHC) solvers often struggle
to discover deep counterexamples
whose derivations may require extensive rule unfolding.
Although large language models (LLMs) can propose long-range symbolic relations,
prior work on formally checked LLM-assisted verification
has mainly focused on generating inductive invariants and concrete traces.
We exploit this broader capability through Derivation
Power Under-approximations (DPUs). Each DPU is a solver-certified
under-approximation of the relation induced by a multi-step CHC derivation.
An LLM proposes each DPU with
symbolic intermediate-state witnesses that enable SMT validation
without introducing additional quantifiers. Validated DPUs can be composed
and reused soundly. Building on this mechanism, our LLM-assisted
framework HornCraft combines invariant generation for over-approximation
with DPU generation for under-approximation in CHC solving.
On 260 linear and nonlinear integer-arithmetic CHC-COMP benchmarks left
unresolved by a 30-minute Spacer run, HornCraft returns certified verdicts
for 228 instances. It certifies 42 instances as unsafe, compared with 22
when DPU construction and reuse are disabled.
\end{abstract}

%% file: sections/1-Introduction.tex
\section{Introduction}

Constrained Horn clauses (CHCs) provide a general intermediate representation
for a wide range of formal verification problems across programming languages and application
domains, including tools for C~\cite{CAV-2015-Seahorn,VMCAI-2022-Korn,CAV-2026-DrvHorn},
Java~\cite{CAV-2016-JayHorn}, Android~\cite{EuroSP-2016-HornDroid},
Rust~\cite{TOPLAS-2021-RustHorn}, and Solidity~\cite{VMCAI-2022-SmartACE,CAV-2022-SolCMC}.
Many solvers have been developed for the CHC satisfiability problem,
including Spacer~\cite{CAV-2014-Spacer}, MuCyc~\cite{PLDI-2024-MuCyc},
Eldarica~\cite{FMCAD-2018-Eldarica}, Golem~\cite{CAV-2023-Golem},
and LoAT~\cite{SAS-2023-ADCL}.
A CHC solver can prove query unreachability by finding an inductive
interpretation, or prove query reachability by applying CHC rules until a
query is reached. Such a finite proof is called a derivation.

For reachable instances, a counterexample derivation may be
deep, requiring many rule applications before reaching
the query~\cite{TACAS-2022-TPA,SAS-2023-ADCL}.
A solver that reasons one rule application at a time has to
traverse the whole derivation,
making deep counterexamples hard to
find and validate. Besides, a derivation may also
branch because one rule application
can depend on multiple reachable predecessors.
Reasoning across several steps at once can avoid this repeated unfolding.
TPA does so by building over-approximate relations for increasingly long
transition sequences~\cite{TACAS-2022-TPA}, whereas LoAT learns
under-approximate relations~\cite{SAS-2023-ADCL}.

On the other hand, LLM-assisted verification has recently been applied to both
proof construction and bug finding.
We follow the guess-and-check paradigm in which
an LLM proposes candidate artifacts and a formal
backend validates them, keeping the LLM outside
the trusted reasoning core.
Prior work in this paradigm has mainly used LLMs to construct
over-approximations of reachable states,
especially inductive invariants, for
unreachability proofs~\cite{FMCAD-2024-Loopy,
ICLR-2024-LEMUR,ASE-2024-LaM4Inv,
ASE-2024-ESBMC-AI,arXiv-2026-CILL,FM-2026-LLM4Ind}.
On the reachability side, recent approaches 
produce concrete inputs or
execution traces as evidence of reachability~\cite{arXiv-2026-Schwarz,
arXiv-2026-AProver}.

However, LLMs are well suited to proposing long-range symbolic relations
that cross many derivation steps. Yet existing guess-and-check interfaces
still require a complete trace for solver validation and thus lack a
direct way to use this ability.

To bridge this gap, we propose Derivation Power Under-approximation (DPU).
A DPU represents a solver-certified subset of feasible multi-step CHC
derivations. It can reduce the number of steps needed to represent a deep
derivation. To obtain this compression within the guess-and-check paradigm, the LLM
proposes candidate DPUs together with symbolic witnesses for their
intermediate states. A solver then validates each proposal through SMT queries
without introducing additional quantifiers. Once certified, DPUs can be
reused to construct new DPUs and combined to
compress longer linear and branching derivations.
In this way, DPUs serve as a formally checkable interface
between an untrusted generative proposer and a trusted verifier.

We develop HornCraft, an LLM-assisted framework for CHC solving.
HornCraft combines LLM-proposed DPU candidates for under-approximating
feasible multi-step derivations with LLM-proposed invariant candidates
for over-approximating reachable states.

Experimental results show that HornCraft certifies 228 of 260 (88\%) linear
and nonlinear integer-arithmetic CHC-COMP instances that time out under a
30-minute Spacer run in Z3 4.17~\cite{TACAS-2008-Z3}. This includes all 35 instances
in the expected-unsafe group and 26 of 40 instances with unknown expected verdicts.
On six parameterized CHC families used to study scalability with derivation
depth and branching, it solves 30 of 37 instances left unresolved by a
30-minute Spacer run. With DPU, HornCraft certifies 42 Spacer-hard instances
as unsafe (vs. 22 without DPU) and solves larger unsafe instances in four
nonlinear families.

Our main contributions are as follows:

1. We introduce witness-carrying DPUs, solver-certified
   under-approximations of multi-step CHC derivations
   that connect LLM proposals with solver validation.
   Symbolic witnesses enable certification without
   introducing additional quantifiers.

2. We give a sound calculus for composing and reusing DPU relations
   over mixed rule applications and branching derivations.

3. We develop HornCraft, an LLM-assisted framework that combines
   LLM-proposed DPU and invariant candidates for reachability
   and unreachability.

%% file: sections/2-Overview.tex
\section{Overview}
\label{sec:overview}

\subsection{Motivating Example}
\label{sec:motivating-example}
Figure~\ref{fig:motivating-example}(a) gives a representation of
the verification task encoded by the CHC-COMP 2026~\cite{CHC-COMP-2026-Results} benchmark
\texttt{s\_split\_15\_000\_001} in Figure~\ref{fig:motivating-example}(b).
The example features a multi-phase loop. During the first 500 iterations,
$p<500$ and $x$ increases by two at each step, reaching 1,000; we call this the
$L$ phase. During the next 500 iterations, $p\geq 500$ and $x$ remains fixed;
we call this the $R$ phase. After the $R$ phase, $p=0$ and
$i=x=1000$, violating the assertion.

\begin{figure}[ht]
  \centering
  \definecolor{DPUBlue}{RGB}{20,91,150}
  \definecolor{CodeKeywordGreen}{RGB}{0,122,92}
  \definecolor{PhaseVar}{RGB}{168,64,130}
  \definecolor{IterVar}{RGB}{0,105,170}
  \definecolor{ValueVar}{RGB}{224,132,0}
  \definecolor{AssertRed}{RGB}{190,45,55}
  \newsavebox{\motivatingProgramBox}
  \newsavebox{\motivatingCHCBox}
  \setlength{\fboxsep}{4pt}
  \begin{lrbox}{\motivatingProgramBox}
    \begin{minipage}{\dimexpr0.47\linewidth-2\fboxsep-2\fboxrule\relax}
      \begin{lstlisting}[
      language=C,
      basicstyle=\ttfamily\scriptsize,
      morekeywords={true},
      keywordstyle=\color{CodeKeywordGreen}\bfseries,
      emph={[1]p},
      emphstyle={[1]\color{PhaseVar}\bfseries},
      emph={[2]i},
      emphstyle={[2]\color{IterVar}\bfseries},
      emph={[3]x},
      emphstyle={[3]\color{ValueVar}\bfseries},
      emph={[4]assert},
      emphstyle={[4]\color{AssertRed}\bfseries},
      numbers=left,
      numberstyle=\tiny\color{black!50},
      numberblanklines=false,
      numbersep=4pt,
      showstringspaces=false,
      frame=none,
      xleftmargin=14pt,
      xrightmargin=0pt,
      aboveskip=0pt,
      belowskip=0pt]
  int p = 0, i = 0, x = 0;
  while (true) {
    assert(!(p == 0 && i == x && i > 0));
    if (p < 500) { x += 2; }
    p = (p + 1) % 1000;
    i++;
  }
\end{lstlisting}
    \end{minipage}
  \end{lrbox}
  \begin{lrbox}{\motivatingCHCBox}
    \begin{minipage}{\dimexpr0.47\linewidth-2\fboxsep-2\fboxrule\relax}
          \centering\scriptsize
          $\begin{aligned}
            \varphi_{r_t} \equiv {}&
              {\color{PhaseVar}p'}=({\color{PhaseVar}p}+1)\bmod 1000
              \land {\color{IterVar}i'}={\color{IterVar}i}+1\\[-0.4mm]
              &{}\land {\color{ValueVar}x'}=\mathsf{ite}({\color{PhaseVar}p}<500,
              {\color{ValueVar}x}+2,{\color{ValueVar}x}),\\[0.6mm]
            r_i:\quad &{\color{PhaseVar}p}=0\land {\color{IterVar}i}=0
              \land {\color{ValueVar}x}=0\Longrightarrow
              \mathit{Inv}({\color{PhaseVar}p},{\color{IterVar}i},{\color{ValueVar}x}),\\[0.6mm]
            r_t:\quad &\varphi_{r_t}\land
              \mathit{Inv}({\color{PhaseVar}p},{\color{IterVar}i},{\color{ValueVar}x})
              \Longrightarrow
              \mathit{Inv}({\color{PhaseVar}p'},{\color{IterVar}i'},
              {\color{ValueVar}x'}),\\[0.6mm]
            r_q:\quad &{\color{PhaseVar}p}=0\land {\color{IterVar}i}={\color{ValueVar}x}
              \land {\color{IterVar}i}>0\land
              \mathit{Inv}({\color{PhaseVar}p},{\color{IterVar}i},{\color{ValueVar}x})
              \Longrightarrow\bot.
          \end{aligned}$
    \end{minipage}
  \end{lrbox}
  \begin{minipage}[t]{0.49\linewidth}
    \vspace{0pt}
    \centering
    \fcolorbox{black!45}{DPUBlue!2}{%
      \usebox{\motivatingProgramBox}\vphantom{\usebox{\motivatingCHCBox}}}
    \par\vspace{3pt}
    {\footnotesize\textbf{(a)}~Verification task}\par
  \end{minipage}
  \hfill
  \begin{minipage}[t]{0.49\linewidth}
    \vspace{0pt}
    \centering
    \fcolorbox{black!45}{DPUBlue!2}{%
      \usebox{\motivatingCHCBox}\vphantom{\usebox{\motivatingProgramBox}}}
    \par\vspace{3pt}
    {\footnotesize\textbf{(b)}~CHC encoding}\par
  \end{minipage}
  \caption{A multi-phase loop and its CHC encoding.}
  \Description{Two side-by-side panels show a loop program on the left and
    its CHC encoding on the right. The encoding consists of an initial fact,
    a transition rule, and a query detecting an assertion violation.}
  \label{fig:motivating-example}
\end{figure}

\begin{figure}[ht]
  \centering
  \includegraphics[width=0.72\linewidth]{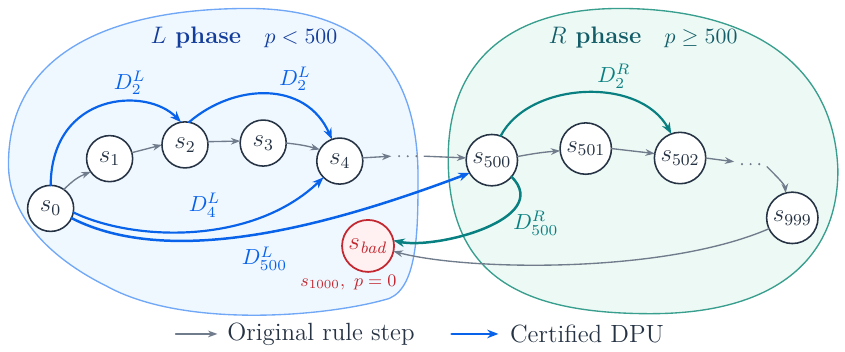}
  \caption{DPU compression in the motivating example.}
  \Description{A two-phase derivation connects the initial state to the bad
    state after 1,000 original transitions. Certified DPUs summarize two,
    four, and 500 transitions, reducing the number of represented steps.}
  \label{fig:linear-dpu-witness}
\end{figure}

Figure~\ref{fig:linear-dpu-witness} shows the 1,000-step counterexample
from the initial state \(s_0\) to the violating state \(s_{1000}\).
Explicitly unfolding it requires 1,000 applications of \(r_t\).
A stepwise search must propagate reachability through all 1,000 applications
of \(r_t\) before the query rule \(r_q\) can be triggered.

We let an LLM propose a DPU, namely a restricted endpoint relation for a
finite CHC derivation, together with witnesses for its hidden states.
An SMT check without new quantifiers certifies that every admitted
endpoint pair expands to a valid sequence of original-rule steps.
Thus, the DPU can compress reachability proofs while preserving soundness.
It may omit valid paths but cannot introduce spurious ones.

For each of the two 500-iteration phases described above, we let the LLM propose
how the state changes across the whole phase instead of listing every
application of $r_t$. Let $s=(p,i,x)$ and $s'=(p',i',x')$ denote the states
before and after a phase, respectively. This gives the following candidate
relations:
\begin{equation}
\begin{aligned}
&D_{500}^L(s,s') \equiv p=0 \land p'=500 \land i'=i+500 \land x'=x+1000,\\
&D_{500}^R(s,s') \equiv p=500 \land p'=0 \land i'=i+500 \land x'=x.
\end{aligned}
\label{eq:motivating-500-step-dpus}
\end{equation}

Assume that both candidates have been certified,
as shown in Fig.~\ref{fig:linear-dpu-witness}.
The following formula checks whether
they connect the initial state to the bad state:
\begin{equation}
s_0=(0,0,0) \land D_{500}^L(s_0,s_{500}) \land D_{500}^R(s_{500},s_{1000})
\land s_{1000}=(0,1000,1000).
\end{equation}

A candidate relation proposed by the LLM can be used as a DPU only after it passes the solver-backed
certification check. Once the two 500-step DPUs have been certified, the two-step SMT
query above can yield a witness to the assertion violation without unfolding 1,000
individual applications of \(r_t\) in the final reachability query. However,
certifying either 500-step relation by unfolding 500 copies of $r_t$ is still a challenge.
Following prior work~\cite{TACAS-2022-TPA}, the framework can first certify shorter DPUs,
such as those spanning 2, 4, or 8 steps (the 2- and 4-step cases are illustrated in
Fig.~\ref{fig:linear-dpu-witness}), and reuse them to certify the 500-step candidates. 
More generally, the framework can also support branching derivations.

\subsection{Framework}

Figure~\ref{fig:framework-overview} summarizes HornCraft. HornCraft first runs
Spacer on the input CHCs for a short time. This initial run reflects the
practical observation that Spacer often resolves instances it can solve
quickly. For unresolved instances, the LLM inspects the CHCs and the certified
banks, then proposes invariant atoms, witness-carrying DPUs,
and reach facts (RFs) with derivations.
The formal engine checks these candidates and allows only those that
pass its checks to cross the trusted boundary and enter the certified banks.

HornCraft then uses the certified banks to search for proofs of query
unreachability and reachability. 
On the reachability side, a DPU
under-approximates the relation induced by a derivation fragment.
The witness check certifies the proposed relation using the supplied
symbolic witnesses.
Certified DPUs can be reused to construct further
DPUs and derive RFs, which record concrete reachable states.
The BMC-like check validates each proposed RF by checking a finite
derivation built from original rules, certified DPUs, and previously
certified RFs. HornCraft uses certified RFs to extend the reachability
search and proves unsafety when they establish reachability of an
original query.
On the unreachability side, Houdini~\cite{FME-2001-Houdini} filters invariant
atoms to obtain inductive invariants that over-approximate reachable states.
HornCraft proves safety when these invariants exclude every original query.

The LLM revises its proposals using the certified banks and verification
feedback, and can request guided Spacer calls using certified invariants
and RFs. If the shared resource budget is exhausted without a proof,
HornCraft returns \texttt{unknown}.
Section~\ref{sec:bi-approximate-solving} details the algorithm.

\begin{figure}[H]
  \centering
  \includegraphics[width=\linewidth]{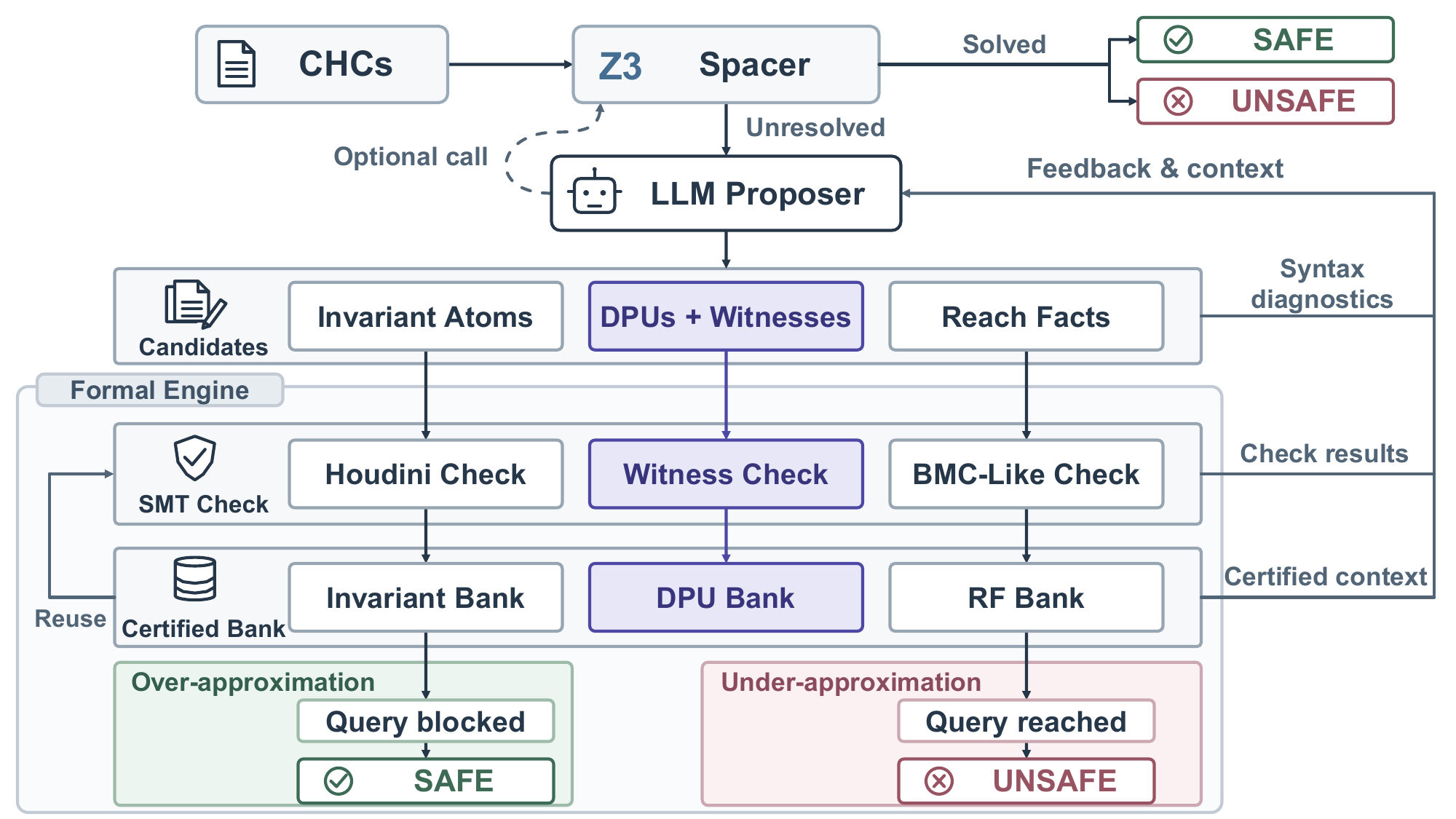}
  \caption{Overview of the HornCraft framework.}
  \Description{HornCraft first runs Spacer on the input CHCs. For unresolved
    instances, an LLM proposes invariant atoms, witness-carrying DPUs, and reach
    facts with derivations. Solver-backed checks admit certified artifacts across
    the trusted boundary for over-approximate safety and under-approximate
    reachability reasoning.}
  \label{fig:framework-overview}
\end{figure}

%% file: sections/3-Preliminaries.tex
\section{Preliminaries}
\label{sec:preliminaries}
We introduce CHCs and their satisfiability problem, whose satisfiable and
unsatisfiable cases are certified by inductive interpretations and finite
CHC derivations, respectively.

\textbf{Constrained Horn clauses (CHCs).}
CHCs are a fragment of first-order logic modulo background theories.
Let $\mathcal{T}$ be a first-order theory and
$\mathcal{P}=\{P_1,\ldots,P_n\}$ be a finite set of uninterpreted relation symbols.
A CHC system $\mathcal{C}$ over $\mathcal{T}$ and $\mathcal{P}$ is a finite set of CHC rules.
Each rule $r\in\mathcal{C}$ is a universally quantified formula
\begin{equation}
    r \;\equiv\;
    \forall \bar{x}_r.\;
    \left(
        \varphi_r
        \land Q_1(\bar{t}_1)
        \land \cdots
        \land Q_m(\bar{t}_m)
        \Longrightarrow H_r
    \right).
    \label{eq:chc}
\end{equation}

Here $m \ge 0$, and $\bar{x}_r$ denotes all variables occurring in $r$.
The quantifier-free $\mathcal{T}$-formula $\varphi_r$ is called the \emph{constraint} of $r$
and contains no relation symbols from $\mathcal{P}$.
For every $1\leq j\leq m$, $Q_j(\bar{t}_j)$ is a relation atom with $Q_j\in\mathcal{P}$,
where $\bar{t}_j$ is a tuple of $\mathcal{T}$-terms matching the arity of $Q_j$.
The relation symbols $Q_1,\ldots,Q_m$ need not be distinct.
$H_r$ is either a relation atom $P_i(\bar{t}_0)$ with $P_i\in\mathcal{P}$ or $\bot$.
We call the antecedent of $r$ its \emph{body} and $H_r$ its \emph{head}.

A rule with no relation atom in its body, i.e., $m=0$, and
one relation atom in its head is called a \emph{fact}.
A rule whose head is $\bot$ is called a \emph{query}.

A rule is \emph{linear} if $m\leq 1$, and \emph{nonlinear} otherwise.
$\mathcal{C}$ is \emph{linear} if every rule in
$\mathcal{C}$ is linear, and \emph{nonlinear} otherwise.
We omit leading universal quantifiers from CHC rules and the background theory $\mathcal{T}$
from the notation when clear from context.

The \emph{CHC satisfiability problem} asks whether $\mathcal{C}$ admits
an interpretation of the relation symbols in $\mathcal{P}$ that satisfies every rule in $\mathcal{C}$.

\textbf{CHC satisfiability as a query-reachability problem.}
For $\mathcal{C}$, this problem can be viewed as a query-reachability problem
over the state spaces associated with its predicates.
For each predicate $P_i$, the state space $\mathcal{S}_{P_i}$ is the Cartesian product of its argument domains.
A tuple $\bar a\in\mathcal{S}_{P_i}$ represents a state associated with $P_i$.
A valuation $\nu$ for a rule $r$ assigns values to the variables in $\bar{x}_r$ and extends to terms in the usual way;
we write $\nu\models\varphi_r$ when the constraint holds under $\nu$.
Facts with head $P_i(\bar t_0)$ generate initial states at $P_i$, whereas other rules with relation-atom heads
propagate reachability from the states represented by their body atoms to a state represented by their head.
For a query with body atoms $Q_1(\bar t_1),\ldots,Q_m(\bar t_m)$, its reachability target is defined
over the product $\mathcal{S}_{Q_1}\times\cdots\times\mathcal{S}_{Q_m}$ of the corresponding state spaces.

For each $P_i\in\mathcal{P}$, let $\mathit{Init}_{P_i}\subseteq\mathcal{S}_{P_i}$ be the set of all states
$\nu(\bar t_0)$ generated by any fact rule $r\in\mathcal{C}$ of the form
$\varphi_r\Longrightarrow P_i(\bar t_0)$, under any valuation $\nu$ for $r$ such that $\nu\models\varphi_r$.
The reachable-state sets $\mathcal{R}_{P_i}\subseteq\mathcal{S}_{P_i}$ form the least family that contains
every $\mathit{Init}_{P_i}$ and is closed under all non-fact, non-query rules.
In particular, consider such a rule $r$ with head $P_i(\bar t_0)$ and body atoms
$Q_1(\bar t_1),\ldots,Q_m(\bar t_m)$.
For every valuation $\nu$ for $r$, if
$\nu\models\varphi_r$ and
$\forall\,1\leq j\leq m.\;
\nu(\bar t_j)\in\mathcal{R}_{Q_j}$,
then $\nu(\bar t_0)\in\mathcal{R}_{P_i}$.

A query is \emph{reachable} if some valuation satisfies its constraint
and instantiates every body atom with a reachable state.
Consequently, $\mathcal{C}$ is \emph{unsafe} (unsatisfiable) iff some
query is reachable; otherwise it is \emph{safe} (satisfiable).

\textbf{Certificates for the query-reachability problem.}
Unreachability certificates use over-approxi\-mations of the reachable-state set
of each predicate, while reachability certificates use under-approximations.

\emph{Unreachability certificate.}
An inductive interpretation $\mathcal{I}$ assigns to every relation symbol
$P_i\in\mathcal{P}$ a background-theory formula $I_{P_i}(\bar x_{P_i})$ of the same arity.
For an atom $A=P_i(\bar t)$, write $\mathcal{I}(A)=I_{P_i}(\bar t)$, and let $\mathcal{I}(\bot)=\bot$.
It is inductive for $\mathcal{C}$ if, for every rule $r\in\mathcal{C}$,
\begin{equation}
    \models
    \varphi_r
    \land
    \bigwedge_{j=1}^{m} I_{Q_j}(\bar{t}_j)
    \Longrightarrow
    \mathcal{I}(H_r).
    \label{eq:inductive-interpretation}
\end{equation}

The formulas $I_{P_i}$ form a family of inductive invariants, one for each relation symbol
$P_i\in\mathcal{P}$.
For every $P_i\in\mathcal{P}$, the set of states satisfying $I_{P_i}$ over-approximates $\mathcal{R}_{P_i}$.
Together, these invariants prove every query unreachable.

\emph{Reachability certificate.}
A CHC derivation is a finite tree of rule applications.
An application of a rule $r$ under a valuation $\nu\models\varphi_r$
derives its instantiated head from its instantiated body atoms.
These body atoms are established at its child nodes in the derivation tree.
Fact applications are leaves of the derivation tree.
A derivation $\delta$ rooted at a query application concluding $\bot$
is a \emph{reachability certificate}. For example, consider the following CHCs:
\begin{equation}
\begin{aligned}
    &x=0\Longrightarrow P_1(x),\quad
    &P_1(x)\land x'=x+1\Longrightarrow P_1(x'),\quad
    &P_1(x)\land x=2\Longrightarrow\bot.
\end{aligned}
\end{equation}

The fact derives $P_1(0)$. Then two applications of the second rule
can derive $P_1(1)$ and $P_1(2)$.
Finally, the query can derive $\bot$ using $P_1(2)$.
These applications form one derivation.
As this example shows, linear CHC derivations are paths,
while nonlinear CHC derivations may branch.

In the derivation $\delta$, the fact leaves and all intermediate nodes derive reachable states.
The root query is therefore reachable, proving $\mathcal{C}$ unsatisfiable.

%% file: sections/4-Compressing-CHC-Derivations.tex
\section{Compressing CHC Derivations}
\label{sec:derivation-compression}

%% file: sections/4.1-Derivation-Power-Under-approximations.tex
\subsection{Derivation Power Under-approximations}
\label{sec:chc-derivations}

A reachability certificate can be highly redundant when every state in its
derivation tree is represented explicitly. It is therefore reasonable to reduce the number of
intermediate states by using sound CHC rules derived from the original $\mathcal{C}$.
For the CHC system in Fig.~\ref{fig:motivating-example}, a certified $D_{500}^L$ compresses
500 applications of $r_t$ into one such rule.

Consider two CHC rules in the form of Eq.~\eqref{eq:chc}, with their variables
renamed apart. If the head of the first rule matches a selected body-atom
occurrence of the second, we compose the two rules by replacing that occurrence
with the body atoms of the first rule and equating the argument tuples of the
matched atoms. The composed rule has the head of the second rule.
For a rule $r$, let $B(r)$ denote its multiset of body atoms and let $H(r)$
denote its head.

We can compose the rules in $\bar r=(r_1,\ldots,r_n)$ in order. At each step
$i>1$, we select a body-atom occurrence $a_i$ of $r_i$ with the same predicate
symbol as the head of the rule composed so far. These selections define the composition plan
$\pi=(a_2,\ldots,a_n)$.

Let $F_i=(B_i^\pi,H_i^\pi,\varphi_i^\pi)$ denote the remaining body atoms, head, and
constraint after composing the first $i$ rules according to $\pi$.
For the first rule,
$F_1=(B(r_1),H(r_1),\varphi_{r_1})$. For $i>1$, let $e_i$ equate the
arguments of $H_{i-1}^\pi$ and $a_i$. Then
\begin{equation}
\begin{aligned}
  B_i^\pi &\;=\; B_{i-1}^\pi\;\uplus\;
      \bigl(B(r_i)\setminus\{a_i\}\bigr),
  & H_i^\pi &\;=\; H(r_i),
  & \varphi_i^\pi \;\equiv\;
      \varphi_{i-1}^\pi\land\varphi_{r_i}\land e_i.
\end{aligned}
\label{eq:composition-interface}
\end{equation}

Here $\uplus$ denotes multiset addition. We write
$B_\pi=B_n^\pi$ and $H_\pi=H_n^\pi$. Let $\bar v$ collect the variables
occurring in $B_\pi$ or $H_\pi$, and let $\bar u$ collect the remaining variables
of $\varphi_n^\pi$. To express the composition over $\bar v$, we existentially
quantify $\bar u$ and obtain the exact composed constraint:
\begin{equation}
  \varphi_\pi(\bar v)\;\equiv\;
  \exists\bar u.\;\varphi_n^\pi(\bar v,\bar u).
\label{eq:composed-fragment-relation}
\end{equation}

Obtaining an equivalent quantifier-free constraint may require expensive
quantifier elimination. To construct a reachability certificate, we can use
sound derived rules that cover selected feasible derivations. We therefore
seek a \emph{Derivation Power Under-approximation (DPU)} of $\varphi_\pi$.
\begin{definition}[Derivation Power Under-approximation]
\label{def:general-dpu}
For a composition plan $\pi$, a relation $D(\bar v)$ is a DPU if every
assignment to $\bar v$ satisfying $D$ also satisfies $\varphi_\pi$.
We write $D\subseteq\varphi_\pi$ for this inclusion.
\end{definition}
A DPU $D$ for the composition plan $\pi$ induces the following CHC rule $r_D$, 
which can be used to compress reachability certificates:
\begin{equation}
  r_D:\quad D(\bar v)\land\bigwedge_{A\in B_\pi}A
       \Longrightarrow H_\pi.
\label{eq:general-dpu}
\end{equation}

%% file: sections/4.2-Witness-Carrying-DPUs-for-Linear-CHCs.tex
\subsection{Witness-Carrying DPUs for Linear CHCs}
\label{sec:linear-dpu}

To certify a candidate relation $D$ proposed by the LLM to be a valid DPU,
we should ensure the inclusion $D\subseteq\varphi_\pi$ in Definition~\ref{def:general-dpu}.
This requires showing that every assignment admitted by $D$ can be extended
with values for the internal variables that satisfy the composed constraints.
We first illustrate this requirement with two composable linear CHC rules:
\begin{equation}
    r_1:\quad P_1(\bar{x})\land\varphi_1(\bar{x},\bar{y})
      \Longrightarrow P_2(\bar{y}),\qquad
    r_2:\quad P_2(\bar{y}')\land\varphi_2(\bar{y}',\bar{z})
      \Longrightarrow P_3(\bar{z}).
\label{eq:linear-rule-composition-example}
\end{equation}

Their exact composed constraint is:
\begin{equation}
  \varphi_{12}(\bar{x},\bar{z})\;\equiv\;
  \exists\bar{y},\bar{y}'.\;
  \bigl(\varphi_1(\bar{x},\bar{y})\land\bar{y}=\bar{y}'
  \land\varphi_2(\bar{y}',\bar{z})\bigr).
\label{eq:linear-composed-constraint}
\end{equation}

For this two-rule composition, we write $D_{12}(\bar x,\bar z)$ for the
candidate relation $D$ and $r_1\circ r_2$ for the exact composed rule
with constraint $\varphi_{12}$.
Since the equality in Eq.~\eqref{eq:linear-composed-constraint}
identifies $\bar y$ and $\bar y'$, we denote their shared value by $\bar u$.
By Definition~\ref{def:general-dpu}, the DPU condition is
\begin{equation}
  \forall \bar{x}, \bar{z}.\;\bigl(
      D_{12}(\bar{x}, \bar{z}) \Longrightarrow \exists \bar{u}.\;
    \bigl(
        \varphi_1(\bar{x},\bar{u}) \land\varphi_2(\bar{u},\bar{z})
    \bigr)
  \bigr).
  \label{eq:linear-dpu-inclusion}
\end{equation}

Checking the implication in Eq.~\eqref{eq:linear-dpu-inclusion}
directly requires reasoning about the
existentially quantified intermediate state $\bar{u}$.
We want to avoid quantifier elimination to keep certification scalable.
Since \(D_{12}\) is intended to describe selected endpoint
pairs of the exact two-rule composition,
we ask the LLM to also provide a symbolic Skolem witness
\(w_{12}(\bar{x},\bar{z})\) for their intermediate states.
\(w_{12}(\bar{x},\bar{z})\) is a tuple of background-theory terms
specifying an intermediate state from the two given endpoints.
It therefore suffices to check the following stronger implication:
\begin{equation}
  \forall \bar{x}, \bar{z}.\;\bigl(
    D_{12}(\bar{x}, \bar{z}) \Longrightarrow
    \varphi_1(\bar{x},w_{12}(\bar{x}, \bar{z}))
    \land\varphi_2(w_{12}(\bar{x}, \bar{z}),\bar{z})
  \bigr).
\label{eq:linear-dpu-check-valid}
\end{equation}

The witness-based check extends directly to the composition plans
defined in Section~\ref{sec:chc-derivations}. For a plan $\pi$,
the LLM supplies a tuple of background-theory terms $w(\bar v)$
that assigns values to all internal variables $\bar u$.
Substituting these terms into the accumulated constraint gives
the sufficient condition
$D(\bar v)\Longrightarrow\varphi_n^\pi(\bar v,w(\bar v))$.
We check this condition by testing the following formula for
unsatisfiability without introducing additional quantifiers:
\begin{equation}
  \mathsf{Cert}(\bar r,\pi,D,w)\;\equiv\;
  D(\bar v)\land\neg\varphi_n^\pi(\bar v,w(\bar v)).
\label{eq:dpu-certification}
\end{equation}

An \texttt{unsat} result proves $D\subseteq\varphi_\pi$.
For a chain of $n$ transition rules of the form illustrated above,
the witnesses specify the $n-1$ intermediate states.

\begin{theorem}[Soundness of Witness-Carrying DPUs]
\label{thm:dpu-certification-soundness}
For a composition plan $\pi$ over $\bar r$, suppose
$\mathcal C\models r_i$ for every source rule $r_i$.
If $\mathsf{Cert}(\bar r,\pi,D,w)$ is \textup{\texttt{unsat}}, then the derived rule $r_D$
in Eq.~\eqref{eq:general-dpu} is a logical consequence of $\mathcal C$.
Consequently, adding $r_D$ to $\mathcal C$ preserves satisfiability.
\end{theorem}
\begin{proof}
Let $\mathcal I\models\mathcal C$, so $\mathcal I$ satisfies every source rule.
For any valuation satisfying the body of $r_D$ under $\mathcal I$,
assign the internal variables $\bar u$ the values of $w(\bar v)$.
The \texttt{unsat} condition ensures that every source constraint and matching
equality in $\pi$ holds under this extended valuation.
The open body atoms in $B_\pi$ remain true, so successive applications
of the source rules in plan order establish $H_\pi$.
Thus $\mathcal I\models r_D$. Since $\mathcal I$ was arbitrary,
$\mathcal C\models r_D$, and $\mathcal C$ and $\mathcal C\cup\{r_D\}$
have the same models.
\end{proof}

We now illustrate this certification using the CHC system in
Figure~\ref{fig:motivating-example}.
Let $s=(p,i,x)$ and $s'=(p',i',x')$ denote the source and destination
states, respectively, and write $r^k$ for the $k$-fold self-composition
$r\circ\cdots\circ r$.
For $\sigma\in\{L,R\}$, let $D_k^\sigma(s,s')$
denote a candidate relation intended to satisfy
$D_k^\sigma(s,s')\subseteq\varphi_{r_t}^{\,k}(s,s')$.

For $k=2$, the LLM proposes
the following two candidate relations, each compressing two applications
of $r_t$:
\begin{equation}
  \begin{aligned}
    D_2^L(s,s')\;\equiv\;{}&
      0\leq p\leq 498
      \land p'=p+2
      \land i'=i+2
      \land x'=x+4,\\
    D_2^R(s,s')\;\equiv\;{}&
      500\leq p\leq 998
      \land p'=(p+2)\bmod 1000
      \land i'=i+2
      \land x'=x.
  \end{aligned}
  \label{eq:motivating-d2l-witness}
\end{equation}

The witnesses for the intermediate states given by the LLM are
$w_L(s,s')=(p+1,i+1,x+2)$ and $w_R(s,s')=(p+1,i+1,x)$.
For $w_L$, substituting the witness into the two copies of the original
transition relation gives
\begin{equation}
  \begin{aligned}
    \varphi_{r_t}(s,w_L)\;\equiv{}&
      p+1=(p+1)\bmod1000\land i+1=i+1\\
      &{}\land x+2=\mathsf{ite}(p<500,x+2,x),\\
    \varphi_{r_t}(w_L,s')\;\equiv{}&
      p'=((p+1)+1)\bmod1000\land i'=(i+1)+1\\
      &{}\land x'=\mathsf{ite}(p+1<500,x+4,x+2).
  \end{aligned}
  \label{eq:linear-dpu-witness-substitution}
\end{equation}

Therefore, the solver checks the following formula for unsatisfiability
to verify $D_2^L$ as a DPU:
\begin{equation}
  \mathsf{unsat}\Bigl(
    D_2^L(s,s')\land
    \neg\bigl(
      \varphi_{r_t}(s,w_L(s,s'))
      \land\varphi_{r_t}(w_L(s,s'),s')
    \bigr)
  \Bigr).
  \label{eq:motivating-d2l-checks}
\end{equation}
The certification for $D_2^R$ is analogous.

Semantically, the guard in $D_2^L$ ensures
that both the source value
$p$ and the intermediate value $p+1$ remain below $500$.
The witness check therefore establishes
$D_2^L\subseteq\varphi_{r_t}^{\,2}$ without
explicitly quantifying over the intermediate state.

For $D_4^L$, the LLM proposes the following candidate relation:
\begin{equation}
  D_4^L(s,s')\;\equiv\;
  0\leq p\leq 496\land p'=p+4\land i'=i+4\land x'=x+8.
  \label{eq:motivating-d4l}
\end{equation}

The LLM provides three witnesses
$w_1(s,s')=(p+1,i+1,x+2)$, $w_2(s,s')=(p+2,i+2,x+4)$, and
$w_3(s,s')=(p+3,i+3,x+6)$ for the hidden intermediate states.
The guard $0\leq p\leq496$ keeps all four applications within the $L$ phase,
so an \texttt{unsat} result for the witness check establishes
$D_4^L\subseteq\varphi_{r_t}^{\,4}$.

In the reachability semantics of
Section~\ref{sec:preliminaries}, either certified two-step
DPU can propagate any reachable source state across
both applications:
$ s\in\mathcal{R}_{\mathit{Inv}} \land D_2^\sigma(s,s')
  \Longrightarrow s'\in\mathcal{R}_{\mathit{Inv}}. $
Figure~\ref{fig:linear-dpu-witness} illustrates this correspondence.
Within the $L$ phase, the original 500-step trace
$s_0\to s_1\to\cdots\to s_{500}$ can therefore be represented using
$D_2^L$ as $s_0\to s_2\to s_4\to\cdots\to s_{500}$, yielding a
250-step derivation. The analogous use of $D_2^R$ compresses the $R$ phase
from 500 steps to 250 steps, reducing the complete two-phase derivation to
500 steps. Using the direct four-step candidates $D_4^L$ and $D_4^R$ instead
reduces each phase to 125 steps, yielding a 250-step derivation overall.

%% file: sections/4.3-Compositional-Construction-of-DPUs.tex
\subsection{Compositional Construction of DPUs}
\label{sec:dpu-composition-reuse}

Due to the monotonicity of relational composition,
if $D_1\subseteq\phi_1$ and $D_2\subseteq\phi_2$, where
$\phi_1$ and $\phi_2$ are the corresponding exact endpoint relations, then
$D_1\circ D_2\subseteq\phi_1\circ\phi_2$.
Consequently, certifying $D_3\subseteq D_1\circ D_2$ can prove
$D_3\subseteq\phi_1\circ\phi_2$.

We therefore can reuse previously certified DPUs to certify new ones.
For a composition plan $\pi$ over $\bar r=(r_1,\ldots,r_n)$, each $r_i$
is either an original rule or a rule induced by a previously certified DPU.
Let $r_D$ be the rule in Eq.~\eqref{eq:general-dpu}.
We define $\mathcal C\vdash r$ by finite applications of the following rules:
\begin{equation}
  \frac{r\in\mathcal C}{\mathcal C\vdash r}
  \quad\textsc{Base}
  \qquad
  \frac{
    \mathcal C\vdash r_i\ (1\leq i\leq n)\qquad
    \mathsf{unsat}\bigl(\mathsf{Cert}(\bar r,\pi,D,w)\bigr)
  }{
    \mathcal C\vdash r_D
  }
  \quad\textsc{DPU-Derive}.
\label{eq:dpu-chc-derivation}
\end{equation}

Induction on the finite derivation of $\mathcal C\vdash r$ shows that
$\mathcal C\models r$. The \textsc{Base} case holds because $r\in\mathcal C$.
For \textsc{DPU-Derive}, the induction hypothesis gives $\mathcal C\models r_i$
for every source rule, and Theorem~\ref{thm:dpu-certification-soundness}
then yields $\mathcal C\models r_D$.

\begin{figure}[ht]
  \centering
  \includegraphics[width=0.60\linewidth]{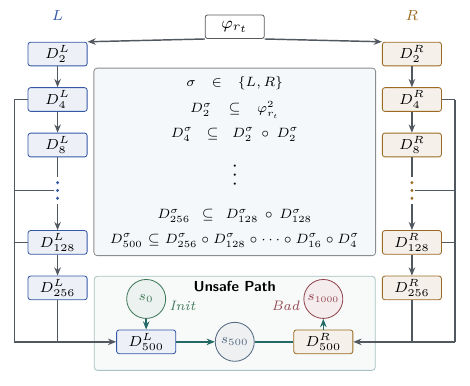}
  \caption{Compositional DPU certification for the motivating example.}
  \Description{Two dependency chains certify increasingly long DPUs for
    the left and right phases. Previously certified powers are combined
    into two 500-step DPUs, which connect the initial state to the bad state.}
  \label{fig:compositional-dpu-certification}
\end{figure}

For the CHC system in Figure~\ref{fig:motivating-example},
Section~\ref{sec:linear-dpu} certified $D_4^L$ directly against four copies
of $r_t$ using three witnesses. We can instead certify the same relation
from two instances of the certified $D_2^L$ in
Eq.~\eqref{eq:motivating-d2l-witness}, using only the midpoint witness
$w(s,s')=(p+2,i+2,x+4)$:
\begin{equation}
  \mathsf{unsat}\Bigl(
    D_4^L(s,s')\land
    \neg\bigl(D_2^L(s,w(s,s'))\land D_2^L(w(s,s'),s')\bigr)
  \Bigr).
  \label{eq:motivating-d4-compositional-check}
\end{equation}

This establishes $D_4^L\subseteq D_2^L\circ D_2^L\subseteq
\varphi_{r_t}^{\,4}$.
We now explain how to repeatedly apply the construction to yield $D_{500}^L$ and
$D_{500}^R$ as shown in Fig.~\ref{fig:compositional-dpu-certification}.
For each $k\in\{2,4,8,\ldots,256\}$, the candidates have the form:
\begin{equation}
  \begin{aligned}
    D_k^L(s,s')\equiv {}& 0\leq p\leq 500-k\land p'=p+k
      \land i'=i+k\land x'=x+2k,\\
    D_k^R(s,s')\equiv {}& 500\leq p\leq 1000-k
      \land p'=(p+k)\bmod 1000 \land i'=i+k\land x'=x.
  \end{aligned}
  \label{eq:motivating-dk-relations}
\end{equation}

The $k=2$ candidates are certified by Eq.~\eqref{eq:motivating-d2l-checks}
and its $R$-phase analogue. For each $k\in\{4,8,\ldots,256\}$, we certify
$D_k^\sigma\subseteq D_{k/2}^\sigma\circ D_{k/2}^\sigma
\subseteq\varphi_{r_t}^{\,k}$ using the symbolic
midpoint witnesses
\begin{equation}
    w_k^L(s,s')=\bigl(p+\frac{k}{2},\,i+\frac{k}{2},\,x+k\bigr), \quad w_k^R(s,s')=\bigl(p+\frac{k}{2},\,i+\frac{k}{2},\,x\bigr).
  \label{eq:motivating-dk-midpoint-witnesses}
\end{equation}

The certified powers are then recombined to obtain $D_{500}^L$ and
$D_{500}^R$, each with five intermediate-state witnesses:
\begin{equation}
    D_{500}^\sigma \subseteq
    D_{256}^\sigma\circ D_{128}^\sigma\circ D_{64}^\sigma
    \circ D_{32}^\sigma\circ D_{16}^\sigma\circ D_4^\sigma
    \subseteq\varphi_{r_t}^{\,500}.
  \label{eq:motivating-d500-composition}
\end{equation}

Together, the certified $D_{500}^L$ and $D_{500}^R$ establish the unsafe
path shown in Figure~\ref{fig:linear-dpu-witness}.

%% file: sections/4.4-DPUs-for-Nonlinear-CHCs-and-Reach-Fact-Reuse.tex
\subsection{DPUs for Nonlinear CHCs and Reach-Fact Reuse}
\label{sec:nonlinear-dpu}
We now consider DPU certification for nonlinear CHCs. The witness-based
check in Eq.~\eqref{eq:dpu-certification} and the \textsc{DPU-Derive} rule
in Eq.~\eqref{eq:dpu-chc-derivation} remain sound for compositions of nonlinear rules.
When the check is \texttt{unsat} and all source rules are logical consequences of
$\mathcal C$, Theorem~\ref{thm:dpu-certification-soundness} still guarantees that
the derived rule is a logical consequence of $\mathcal C$.

Unlike linear DPU compositions, nonlinear compositions can leave multiple
open body atoms. By Eq.~\eqref{eq:composition-interface}, composing two
rules with $m$ and $m'$ body-atom occurrences leaves $m+m'-1$ such atoms
in $B_\pi$. The DPU relates the states at these frontiers to a successor
state, as illustrated by the following two rules:
\begin{equation}
  \begin{aligned}
    r:\;& P_1(a)\land P_2(b)\land\alpha(a,b,p)\Longrightarrow P_3(p),\\
    r':\;& P_3(p')\land P_4(c)\land\beta(p',c,q)\Longrightarrow P_5(q).
  \end{aligned}
  \label{eq:nonlinear-composition-example}
\end{equation}

Composing $r$ and $r'$ with $P_3$ as the pivot imposes $p=p'$, yielding
$\varphi_\pi(a,b,c,q)=\exists p.\,(\alpha(a,b,p)\land\beta(p,c,q))$.
A certified $D\subseteq\varphi_\pi$ induces
$r_D:P_1(a)\land P_2(b)\land P_4(c)\land D(a,b,c,q)\Longrightarrow P_5(q)$.
The resulting rule $r_D$ has three open frontier atoms: $P_1(a)$,
$P_2(b)$, and $P_4(c)$. As composition continues in this way,
the number of open frontiers grows.

To reduce the number of open frontiers,
following GPDR~\cite{SAT-2012-GPDR} and Spacer~\cite{CAV-2014-Spacer},
we record reach facts (RFs).
A point RF $P_i(\bar a)$ records a concrete reachable state
$\bar a\in\mathcal R_{P_i}$ together with its derivation.
Initially, RFs record initial states generated by fact rules.
The LLM can propose a finite derivation from known RFs to a target,
using original rules and certified DPUs. As in bounded model checking,
an SMT solver checks whether the resulting constraints are jointly satisfiable.
A satisfying model establishes the proposed derivation.
If it concludes with a query application deriving $\bot$, the CHC system is unsafe.
Otherwise, the derived relation atom is recorded as a new RF.
A recorded RF $P_i(\bar a)$ can close an open frontier $P_i(\bar t)$
in a nonlinear composition by adding $\bar t=\bar a$ to the constraint.
RFs can also serve as intermediate results and act as initial facts in subsequent derivations.

%% file: sections/5-LLM-Assisted-Bi-Approximate-CHC-Solving.tex
\section{LLM-Assisted Bi-Approximate CHC Solving}
\label{sec:bi-approximate-solving}

We combine witness-carrying DPUs and invariant atoms proposed by an LLM
in a CHC-solving framework called HornCraft, as shown in Figure~\ref{fig:framework-overview}.
Given an input CHC system $\mathcal C$, HornCraft maintains a certified
bank $\mathcal B=(\mathcal D,\mathcal R,\mathcal I)$ of
DPUs, reach facts, and invariants, respectively.
The invariant family $\mathcal I$ satisfies all non-query rules.
The procedures below share a resource budget $\Delta$. We present LLM-guided
reachability search, invariant synthesis, and the overall algorithm below.

%% file: sections/5.1-LLM-Guided-Reachability-Search.tex
\subsection{LLM-Guided Reachability Search}
\label{sec:llm-reachability}

To prove query reachability, HornCraft searches for a valuation that satisfies
a query's constraint and matches each of its body atoms to a certified RF.
The LLM guides this search by proposing DPU candidates
to compress derivations and RF candidates with finite derivations that
may reuse certified DPUs. The LLM bases its proposals on the input CHC
system $\mathcal C$, the certified bank $\mathcal B$, and feedback.
As described in Sections~\ref{sec:dpu-composition-reuse} and~\ref{sec:nonlinear-dpu},
certified DPUs and RFs can be reused in later derivations. The LLM therefore
submits dependent candidates in an ordered batch $\mathcal{K}$ for
certification by Algorithm~\ref{alg:reach-step}.

\begin{algorithm}[htbp]
  \caption{\textsc{ReachStep}: certifying a batch of DPUs and RFs}
  \label{alg:reach-step}
  \small
  \raggedright
  \textbf{Input:} CHC system $\mathcal C$, certified bank $\mathcal B$,
  candidate batch $\mathcal{K}$, budget $\Delta$.\par
  \textbf{Output:} Updated bank $\mathcal B$ and feedback $f$.\par\smallskip
  \renewcommand{\arraystretch}{1.08}
  \begin{tabular}{@{}r@{\hspace{0.8em}}l@{}}
    \algline{ln:reach:feedback-init} & $f\gets\emptyset$ \\
    \algline{ln:reach:loop} & \textbf{for each} $c\in\mathcal{K}$ \textbf{in submission order while} $\Delta$ permits certification \textbf{do} \\
    \algline{ln:reach:ready} & \quad \textbf{if not} $\textsc{Ready}(c,\mathcal C,\mathcal D,\mathcal R)$ \textbf{then} \\
    \algline{ln:reach:reject} & \qquad record rejection in $f$; \textbf{continue} \\
    \algline{ln:reach:ready-end} & \quad \textbf{end if} \\
    \algline{ln:reach:dpu} & \quad \textbf{if} $c$ is a DPU candidate $(\bar r,\pi,D,w)$ \textbf{then} \\
    \algline{ln:reach:compose} & \qquad $(B_\pi,H_\pi,\varphi_n^\pi)\gets\textsc{ComposeRules}(\bar r,\pi)$ \\
    \algline{ln:reach:witness-check} & \qquad $(s,\mu)\gets\textsc{SMT}(D(\bar v)\land\neg\varphi_n^\pi(\bar v,w(\bar v)))$ \\
    \algline{ln:reach:accept-dpu} & \qquad \textbf{if} $s=\texttt{unsat}$ \textbf{then} \\
    \algline{ln:reach:construct-dpu} & \qquad\quad $r_D\gets\bigl(D(\bar v)\land\bigwedge_{A\in B_\pi}A\Longrightarrow H_\pi\bigr)$ \\
    \algline{ln:reach:record-dpu} & \qquad\quad $\textsc{RecordDPU}(\mathcal D,r_D,\bar r,\pi,w)$ \\
    \algline{ln:reach:accept-dpu-end} & \qquad \textbf{end if} \\
    \algline{ln:reach:rf} & \quad \textbf{else} \\
    \algline{ln:reach:encode-rf} & \qquad $\Psi_c\gets\textsc{EncodeRFDerivation}(c,\mathcal C,\mathcal D,\mathcal R)$ \\
    \algline{ln:reach:rf-check} & \qquad $(s,\mu)\gets\textsc{SMT}(\Psi_c)$ \\
    \algline{ln:reach:record-rf} & \qquad \textbf{if} $s=\texttt{sat}$ \textbf{then} $\textsc{RecordRF}(\mathcal R,c,\mu)$ \\
    \algline{ln:reach:candidate-end} & \quad \textbf{end if} \\
    \algline{ln:reach:feedback} & \quad append the check result and available diagnostics to $f$ \\
    \algline{ln:reach:loop-end} & \textbf{end for} \\
    \algline{ln:reach:return} & \textbf{return} $\mathcal B,f$ \\
  \end{tabular}
\end{algorithm}

Each DPU candidate $(\bar r,\pi,D,w)$ specifies source rules, a composition
plan, a candidate relation, and witnesses. \textsc{SMT} checks each formula
within the remaining budget $\Delta$ and returns a status $s$ of \texttt{unsat},
\texttt{unknown}, or \texttt{sat} with a model $\mu$.

The formal engine processes candidates in submission order within the
budget, rejecting those with uncertified dependencies.
A candidate certified earlier can be reused by later candidates
within the same batch (\alglinesref{ln:reach:loop}{ln:reach:ready-end}).

For a DPU candidate, the engine first renames the source variables apart
and then composes the rules according to $\pi$ using Eq.~\eqref{eq:composition-interface}.
It then checks the formula in Eq.~\eqref{eq:dpu-certification} for unsatisfiability
to certify the candidate
(\alglinesref{ln:reach:dpu}{ln:reach:witness-check}). Sources can include certified RFs as ground facts. If the check
returns \texttt{unsat}, the engine adds the certified rule $r_D$ to $\mathcal D$
(\alglinesref{ln:reach:accept-dpu}{ln:reach:record-dpu}). A \texttt{sat} result instead supplies an endpoint valuation
where the proposed witnesses fail. The engine adds this valuation to $f$.

For an RF candidate $c$, the engine encodes its proposed derivation as an
SMT formula $\Psi_c$ (\alglinesref{ln:reach:rf}{ln:reach:rf-check}).
The formula conjoins equalities matching open frontiers to certified RFs,
constraints encoding rule applications, and any target conditions.
On \texttt{sat}, the engine records the concrete root atom $P(\bar a)$ in
$\mathcal R$ together with the satisfying model and the supporting derivation
(\alglineref{ln:reach:record-rf}). On \texttt{unsat}, the engine attempts to
extract an unsat core by tracking the constraints of each derivation step
separately. Any resulting core is added to $f$.
On \texttt{unknown}, the solver's reason is recorded in $f$.

After processing the batch or exhausting the budget, the procedure returns
the updated bank $\mathcal B$ and $f$ for subsequent LLM proposals
(\alglinesref{ln:reach:feedback}{ln:reach:return}).

To close an original query $q$, \textsc{QueryReached} checks whether one
valuation satisfies its constraint and matches every body atom to a
certified RF in $\mathcal R$. A \texttt{sat} result derives $\bot$ and proves
\texttt{unsafe}. The formal engine retains the query valuation and the supporting
RF certificates. Failure to find a matching that passes this check leaves
the query unresolved.

To justify this \texttt{unsafe} verdict, we show that every certified RF
has a finite derivation over the original rules.

\begin{proposition}[Soundness of joint DPU and RF reuse]
\label{prop:joint-dpu-rf-soundness}
Suppose the input bank has a finite acyclic certificate graph, every
input DPU application expands to finitely many original rule applications
with the same head and open premises, and every input RF has a finite
derivation over $\mathcal C$.
Algorithm~\ref{alg:reach-step} preserves these properties.
Consequently, every RF in the updated bank is reachable, and a successful
\textsc{QueryReached} check yields a finite derivation of $\bot$ over the
original rules.
\end{proposition}

\begin{proof}
Induct jointly on DPU and RF admissions in certification order, with the
input bank as the base case. \textsc{Ready} restricts sources to original
rules and previously certified DPUs or RFs, so each admission preserves
acyclicity.

For a DPU, the witness construction in the proof of
Theorem~\ref{thm:dpu-certification-soundness} expands each application into
its source rule applications. By the induction hypothesis, expanding the
DPU and RF sources yields finitely many original rule applications with
the same head and open premises.

For an RF, the satisfying model instantiates its proposed derivation.
Expanding its certified sources by the induction hypothesis yields a finite
derivation of the recorded atom over $\mathcal C$.
Finally, attaching these RF derivations to the query application validated
by \textsc{QueryReached} yields a finite derivation of $\bot$ over $\mathcal C$.
\end{proof}

%% file: sections/5.2-LLM-Guided-Invariant-Synthesis.tex
\subsection{LLM-Guided Invariant Synthesis}
\label{sec:llm-invariants}

To prove query unreachability, HornCraft follows prior work on LLM-guided
invariant generation~\cite{ICLR-2024-LEMUR,ASE-2024-ESBMC-AI,ASE-2024-LaM4Inv,ISSTA-2025-Clause2Inv,FM-2026-LLM4Ind,arXiv-2026-CILL}.
As in Loopy~\cite{FMCAD-2024-Loopy}, we use Houdini~\cite{FME-2001-Houdini}
to filter candidate atoms.

Beyond proposing invariant atoms, the LLM may also request a short guided
Spacer call. HornCraft passes the certified invariants and RFs in the current
bank to Spacer as guidance. Following CIll's reuse of IC3's infinite
frame~\cite{arXiv-2026-CILL}, HornCraft uses Spacer's infinity-level lemmas
to strengthen its invariant bank. If the call leaves the query unresolved,
the formal engine splits the infinity-level lemmas exported by Spacer for
each predicate into invariant atoms and adds these atoms with their solver
evidence to $\mathcal I$ for the next LLM turn.
The LLM proposes a finite set $A_P$ of invariant atoms for each predicate
$P$. We write $A=(A_P)_{P\in\mathcal P}$.

The \textsc{InvariantStep} procedure certifies invariants for proving
query unreachability. Its candidate family $\mathcal J$ consists of
$J_P=I_P\land\bigwedge_{a\in A_P}a$ for each predicate $P$.
For each original non-query rule, the SMT check uses the current candidate
family $\mathcal J$ in the body and tests whether the newly proposed atoms satisfy
the inductive obligations in Eq.~\eqref{eq:inductive-interpretation}.
On \texttt{sat}, it removes the new atoms falsified at the rule's head
and repeats the checks on the weakened family.
If all checks return \texttt{unsat}, $\mathcal J$ replaces $\mathcal I$.
An \texttt{unknown} result or an exhausted budget leaves the previously
certified family unchanged. The procedure returns $\mathcal I$ and
feedback $f$ containing the check results and available diagnostics.

\textsc{QueryExcluded} checks
$\varphi_q\land\bigwedge_j I_{Q_j}(\bar t_j)$ for every original query $q$.
If all checks return \texttt{unsat}, the certified invariants exclude
every query and prove \texttt{safe}.
Models from \texttt{sat} query checks are also recorded in $f$.

%% file: sections/5.3-The-HornCraft-Algorithm.tex
\subsection{The HornCraft Algorithm}

Algorithm~\ref{alg:horncraft} shows HornCraft's overall workflow.
For instances that Spacer does not solve within 10 seconds, HornCraft
starts the LLM-guided search loop. At each turn, the LLM receives the CHCs,
the current bank, and feedback
from previous checks. It may then propose a mixed batch of invariant atoms,
DPU candidates, and RF candidates with finite derivations, or request a
guided Spacer call
(\alglinesref{ln:horncraft:loop}{ln:horncraft:action}).

\begin{algorithm}[htbp]
  \caption{HornCraft's LLM-guided certification loop}
  \label{alg:horncraft}
  \small
  \raggedright
  \textbf{Input:} CHC system $\mathcal C$ and shared resource budget $\Delta$.\par
  \textbf{Output:} \texttt{safe}, \texttt{unsafe}, or \texttt{unknown}.\par\smallskip
  \renewcommand{\arraystretch}{1.08}
  \begin{tabular}{@{}r@{\hspace{0.8em}}l@{}}
    \algline{ln:horncraft:initialize} & $\mathcal B\gets\bigl(\emptyset,\emptyset,(\top)_{P\in\mathcal P}\bigr)$; $f\gets\emptyset$ \\
    \algline{ln:horncraft:initial-spacer-run} & $(v,\mathcal B,f)\gets\textsc{RunSpacer}(\mathcal C,\mathcal B,\Delta)$ \\
    \algline{ln:horncraft:initial-verdict} & \textbf{if} $v\ne\texttt{unknown}$ \textbf{then return} $v$ \\
    \algline{ln:horncraft:loop} & \textbf{while} $\Delta$ permits interaction \textbf{do} \\
    \algline{ln:horncraft:action} & \quad $a\gets\textsc{NextLLMAction}(\mathcal C,\mathcal B,f,\Delta)$ \\
    \algline{ln:horncraft:batch} & \quad \textbf{if} $a$ is a mixed candidate batch $(\mathcal{K},A)$ \textbf{then} \\
    \algline{ln:horncraft:certify-batch} & \qquad $(\mathcal B,f)\gets\textsc{CertifyBatch}(\mathcal C,\mathcal B,\mathcal{K},A,f,\Delta)$ \\
    \algline{ln:horncraft:query-reached} & \qquad \textbf{if} $\textsc{QueryReached}(\mathcal C,\mathcal R,\Delta)$ \textbf{then return} \texttt{unsafe} \\
    \algline{ln:horncraft:query-excluded} & \qquad \textbf{if} $\textsc{QueryExcluded}(\mathcal C,\mathcal I,\Delta)$ \textbf{then return} \texttt{safe} \\
    \algline{ln:horncraft:guided-request} & \quad \textbf{else} \hspace{0.5em} $(v,\mathcal B,f)\gets\textsc{ApplyGuidance}(\mathcal C,\mathcal B,f,a,\Delta)$ \\
    \algline{ln:horncraft:guided-verdict} & \qquad \textbf{if} $v\ne\texttt{unknown}$ \textbf{then return} $v$ \\
    \algline{ln:horncraft:action-end} & \quad \textbf{end if} \\
    \algline{ln:horncraft:loop-end} & \textbf{end while} \\
    \algline{ln:horncraft:return} & \textbf{return} \texttt{unknown} \\
  \end{tabular}
\end{algorithm}

If the action is a mixed batch, \textsc{CertifyBatch} applies
\textsc{ReachStep} to DPU/RF candidates and \textsc{InvariantStep} to
invariant atoms. It records their check results and diagnostics in $f$
(\alglinesref{ln:horncraft:batch}{ln:horncraft:certify-batch}). Either part
of the batch may be empty. The updated bank is then used to check the
original queries. \textsc{QueryReached} returns \texttt{unsafe} when
certified RFs close a query, while \textsc{QueryExcluded} returns
\texttt{safe} when the certified invariants exclude every query
(\alglinesref{ln:horncraft:query-reached}{ln:horncraft:query-excluded}).
If neither check returns a verdict, their feedback guides the next LLM
proposal.

If the LLM requests a guided Spacer call, \textsc{ApplyGuidance} invokes the
call described in Section~\ref{sec:llm-invariants}
(\alglinesref{ln:horncraft:guided-request}{ln:horncraft:guided-verdict}).
The loop continues until a verdict is returned or the interaction budget is
exhausted. If no verdict is returned, HornCraft returns \texttt{unknown}.
\FloatBarrier

%% file: sections/6-Evaluation.tex
\section{Evaluation}
\label{sec:evaluation}

Our evaluation addresses the following research questions.
\par\smallskip
\noindent\textbf{RQ1.} How effective is HornCraft at solving hard
linear and nonlinear CHC instances?\par
\noindent\textbf{RQ2.} How does HornCraft scale with the size of deep and
branching CHC derivations?\par
\noindent\textbf{RQ3.} How much do DPUs contribute to
HornCraft's solving coverage and scalability?
\par\smallskip

\phantomsection
\label{sec:eval-setup}
\noindent\textbf{Experimental Setup.}
We implemented HornCraft in Python, using Z3 4.17.0 for SMT solving
and its Spacer engine for CHC solving. For the LLM-guided search, we
used \texttt{gpt-5.6-sol} through Codex CLI 0.154.0, with reasoning
effort set to \texttt{high}. We ran five benchmark instances concurrently
on a server with two AMD EPYC 7H12 processors and 1 TiB of RAM.

For RQ1, we start from the complete benchmark sets of the CHC-COMP 2026
LIA-Lin and LIA tracks. We select all instances left unresolved by
Spacer (Z3 4.16.0) in the published results~\cite{CHC-COMP-2026-Results},
yielding 220 LIA-Lin and 64 LIA instances. Since Spacer is sensitive
to solver heuristics~\cite{CAV-2020-GSpacer}, we rescreen all 284 candidates
using standalone Spacer with HornCraft's Z3 4.17.0 and solver
configuration.\footnote{HornCraft disables relation-rewriting transformations
to preserve predicates for guidance injection. The rescreening excludes
20 solved instances and 4 execution errors.}
Each screening run has a 30-minute time limit and a 30 GB address-space
limit. We retain only the instances that time out, yielding a
Spacer-hard evaluation set of 202 LIA-Lin and 58 LIA instances.

For RQ2, we manually develop and adapt six parameterized CHC encodings
to evaluate scalability on deep and branching derivations. We instantiate
these encodings at different parameter values to obtain 49 benchmarks,
collectively called the Depth-and-Branching CHC Suite.
Section~\ref{sec:eval-branching} describes their sources, encodings, and
parameter choices.

For RQ3, we disable DPU construction and reuse and rerun HornCraft on
both benchmark sets to evaluate the contribution of DPUs to solving
coverage and scalability.

All configurations use a 30-minute end-to-end time limit per instance
and a 30 GB memory limit. These budget values match
those used in CHC-COMP 2026~\cite{CHC-COMP-2026-Results}.
After solving, we replay the certificates produced by the formal engine
to recheck their proof steps. Replay time is excluded from the 30-minute
time limit.

%% file: sections/6.1-Overall-Effectiveness.tex
\subsection{RQ1: Overall Effectiveness}
\label{sec:eval-effectiveness}
\paragraph{Overall results.}
On the Spacer-hard CHC-COMP benchmarks, HornCraft solves 228 of 260
instances (88\%), including 186 of 202 LIA-Lin instances and 42 of 58
LIA instances. It proves query reachability for all 35 instances with an
expected unsafe verdict. Among the 40 instances with unknown expected
verdicts, HornCraft returns 19 safe and 7 unsafe results.
Table~\ref{tab:rq1-effectiveness} summarizes the results by track and
expected verdict.\footnote{CHC-COMP 2026 labels
\texttt{aeval-unsafe/s\_split\_50\_000.smt2} as safe based solely on
LoAT's result. HornCraft produces
an unsafe certificate that passes replay, so we count this instance
in the expected-unsafe group.}
Tables~\ref{tab:rq1-effectiveness} and~\ref{tab:rq1-joint-hard} report
solved/total (percentage), with S and U denoting safe and unsafe results.
HornCraft thus achieves broad coverage on hard CHCs, including instances
with previously unknown verdicts. The coverage difference between tracks
is most pronounced among instances with unknown expected verdicts.

\begin{table}[H]
  \caption{HornCraft on Spacer-hard CHC-COMP instances.}
  \label{tab:rq1-effectiveness}
  \centering
  \small
  \setlength{\tabcolsep}{4pt}
  \renewcommand{\arraystretch}{1.08}
  \begin{tabular*}{\linewidth}{@{\extracolsep{\fill}}lllll@{}}
    \toprule
    Track & \multicolumn{1}{c}{Expected safe}
      & \multicolumn{1}{c}{Expected unsafe}
      & \multicolumn{1}{c}{Expected unknown}
      & \multicolumn{1}{c@{}}{Overall} \\
    \midrule
    LIA-Lin
      & \evalCount{137/149} (92\%)
      & \evalCount{32/32} (100\%)
      & \evalCount{17/21} (81\%; \evalOutcomes{11}{6})
      & \evalCount{186/202} (92\%) \\
    LIA
      & \evalCount{30/36} (83\%)
      & \evalCount{3/3} (100\%)
      & \evalCount{9/19} (47\%; \evalOutcomes{8}{1})
      & \evalCount{42/58} (72\%) \\
    \midrule
    Total
      & \evalCount{167/185} (90\%)
      & \evalCount{35/35} (100\%)
      & \evalCount{26/40} (65\%; \evalOutcomes{19}{7})
      & \evalCount{228/260} (88\%) \\
    \bottomrule
  \end{tabular*}
\end{table}

\paragraph{Hard instances for other traditional solvers.}
Table~\ref{tab:rq1-joint-hard} reports HornCraft's coverage on instances
unresolved by other CHC-COMP 2026 solvers.\footnote{Results for LoAT, Golem,
and Eldarica are taken from the published CHC-COMP 2026
records~\cite{CHC-COMP-2026-Results}.}
On our Spacer-hard LIA-Lin instances, HornCraft solves 45 of 52 instances
unresolved by LoAT and 103 of 118 unresolved by Golem. On LIA, it solves
20 of 35 instances unresolved by Eldarica.
HornCraft thus provides complementary coverage on instances unresolved
by both Spacer and another competition solver.

\begin{table}[H]
  \caption{HornCraft on instances unresolved by other solvers in CHC-COMP 2026.}
  \label{tab:rq1-joint-hard}
  \centering
  \small
  \setlength{\tabcolsep}{3pt}
  \renewcommand{\arraystretch}{1.08}
  \begin{tabular*}{\linewidth}{@{\extracolsep{\fill}}lrllll@{}}
    \toprule
    Solver and track & Cases & \multicolumn{1}{c}{Expected safe}
      & \multicolumn{1}{c}{Expected unsafe}
      & \multicolumn{1}{c}{Expected unknown}
      & \multicolumn{1}{c@{}}{Overall} \\
    \midrule
    LoAT--LIA-Lin
      & 52 & \evalCount{18/21} (86\%)
      & \evalCount{10/10} (100\%)
      & \evalCount{17/21} (81\%; \evalOutcomes{11}{6})
      & \evalCount{45/52} (87\%) \\
    Golem--LIA-Lin
      & 118 & \evalCount{80/91} (88\%)
      & \evalCount{6/6} (100\%)
      & \evalCount{17/21} (81\%; \evalOutcomes{11}{6})
      & \evalCount{103/118} (87\%) \\
    Eldarica--LIA
      & 35 & \evalCount{10/15} (67\%)
      & \evalCount{1/1} (100\%)
      & \evalCount{9/19} (47\%; \evalOutcomes{8}{1})
      & \evalCount{20/35} (57\%) \\
    \bottomrule
  \end{tabular*}
\end{table}

\paragraph{Token usage.}
Figure~\ref{fig:rq1-token-distribution} reports token use for all 260 runs,
including unresolved cases. Input includes cached tokens and output includes
reasoning tokens. Median input and output counts are 97.3k and 4.5k for
LIA-Lin, and 1.94M and 21.3k for LIA. Median per-instance input cache hit
rates are 85.4\% and 94.2\%, respectively. Cached tokens account for
94.0\% and 96.6\% of total input in the two tracks, or 95.8\% overall.
LIA is more token-intensive, although most input tokens in both tracks
are served from cache.

\begin{figure}[htbp]
  \centering
  \includegraphics[width=0.88\linewidth]{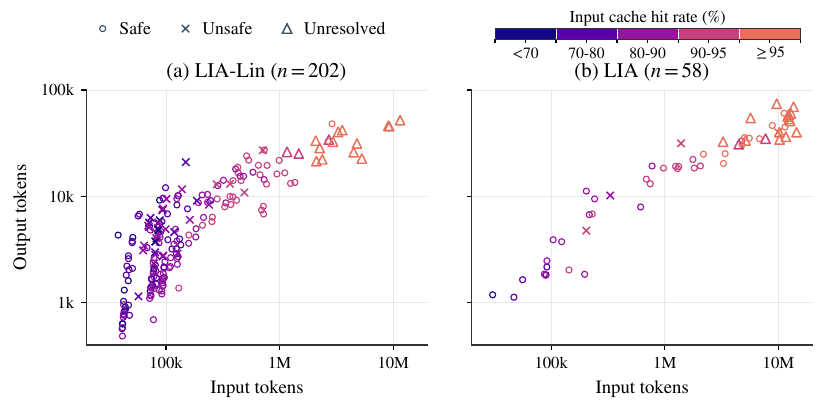}
  \caption{Input and output token use.}
  \label{fig:rq1-token-distribution}
  \Description{Two scatter panels show 202 LIA-Lin executions on the left
    and 58 LIA executions on the right.
    Each point pairs input and output tokens on identical base-10 logarithmic
    axes. Open circles denote safe outcomes, crosses unsafe outcomes,
    and open triangles unresolved outcomes. Five input cache hit bands,
    below 70, 70 to 80, 80 to 90, 90 to 95, and at least 95 percent,
    progress from cool dark blue through purple and magenta to warm orange. LIA-Lin has many
    low-token safe and unsafe runs. Unresolved runs tend toward higher token
    use. Some closely neighboring points overlap.}
\end{figure}

\paragraph{Execution time.}
Figure~\ref{fig:rq1-time-breakdown} breaks down elapsed time for all 260
runs, including unresolved cases. Among solved instances, median
end-to-end times are 105.1 seconds for LIA-Lin and 384.3 seconds for LIA,
both well within the 30-minute budget. Across all runs, median
per-instance Agent shares are 86\% and 88\%, respectively. Agent
interaction thus accounts for most elapsed time in both tracks. In some
long runs, the formal engine also takes up much of the time.
Bars are ordered by total time within each track, excluding certificate
replay. Agent includes communication and waiting. Formal Engine covers
Spacer calls and certificate processing.
The dashed line marks the 30-minute budget.

\begin{figure}[htbp]
  \centering
  \includegraphics[width=\linewidth]{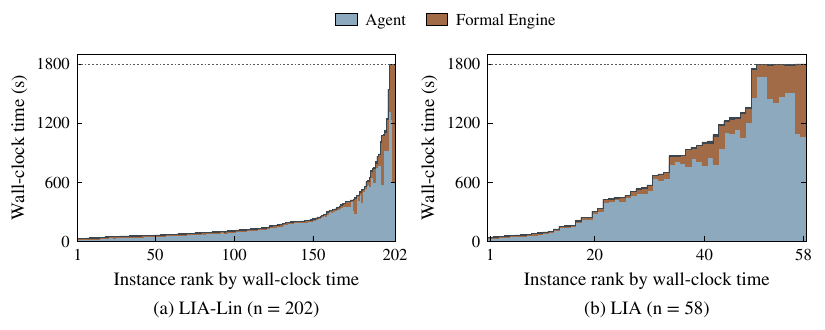}
  \caption{Execution-time breakdown.}
  \label{fig:rq1-time-breakdown}
  \Description{Two columns show all 202 LIA-Lin and 58 LIA Full executions,
    including unresolved runs, sorted within each track by increasing wall-clock
    time. Each stacked bar gives the elapsed seconds of one execution.
    Light blue, ochre, and dark gray distinguish Agent, Formal Engine, and Other,
    respectively, in a fixed bottom-to-top order. Agent interaction accounts
    for most elapsed time in most runs, while some long runs have substantial
    formal-engine processing. The dashed line marks the 1800-second budget.}
\end{figure}
\FloatBarrier

%% file: sections/6.2-Scalability-on-Deep-and-Branching-Derivations.tex
\subsection{RQ2: Scalability on Deep and Branching Derivations}
\label{sec:eval-branching}
\paragraph{Benchmark sources.}
The Depth-and-Branching CHC Suite contains one linear and five nonlinear
families. We adapt Ackermann from existing CHCs in the HOPV benchmark
collection~\cite{HOPVBenchmarks}, preserving its recursive clauses.
The remaining families are our CHC encodings of Hanoi's legal moves,
the Chinese Rings rules~\cite{HeefferHinz2017ChineseRings},
the binomial recurrence~\cite{DLMFBinomialRecurrence} (Pascal), and
Takeuchi's triple recursion~\cite{Knuth1991Recursion} (Tarai).

\paragraph{Encodings and parameters.}
For both Hanoi encodings and Chinese Rings, $n$ denotes the number of
disks and rings, respectively.
The linear encoding derives reachable configurations one move at a time.
The two nonlinear encodings use a predicate that relates two configurations
if one can reach the other through legal moves. Starting from single
moves, a rule builds a longer path by joining two paths when the first
ends where the second begins.
The recursive queries are Ackermann$(3,n) \geq 55$,
$\binom{n}{2}>0$, and Tarai$(n,1,0)>0$.
The nonlinear encodings have at most two predicate premises per rule,
except Tarai, which has four. All solver configurations use the full
suite, including instances solved by Spacer.

\paragraph{Overall results.}
HornCraft solves 42 of the 49 instances, including 30 of the 37 left
unresolved by the 30-minute Spacer baseline.
The solved range extends beyond Spacer's in all six families
(Table~\ref{tab:rq2-families}).
Full is the complete HornCraft solver. Without DPU disables DPU
construction and reuse.
Each configuration also solves every smaller tested $n$.
Tokens include all input and output. Parentheses give the cached fraction of input.

\begin{table}[H]
  \caption{Maximum solved $n$ across benchmark families.}
  \label{tab:rq2-families}
  \centering
  \small
  \setlength{\tabcolsep}{4pt}
  \renewcommand{\arraystretch}{1.08}
  \begin{tabular*}{\linewidth}{@{\extracolsep{\fill}}lrrrrrr@{}}
    \toprule
    & Tested $n$ & \multicolumn{3}{c}{Maximum solved $n$} & \multicolumn{2}{c}{Full at maximum $n$} \\[-1pt]
    \cmidrule(lr){3-5} \cmidrule(lr){6-7}
    \noalign{\vskip -1.5pt}
    Family & (instances) & Spacer & Without DPU & Full & Time (s) & Tokens (M) \\
    \midrule
    Hanoi (linear) & 4--20 (9) & 4 & 12 & 20 & 207.5 & 0.63 (93\%) \\
    Hanoi (path) & 4--20 (9) & 4 & 12 & 20 & 238.3 & 0.59 (91\%) \\
    Ackermann & 3--10 (8) & 4 & 6 & 8 & 841.3 & 3.34 (96\%) \\
    Chinese Rings & 4--20 (9) & 6 & 12 & 20 & 682.5 & 3.24 (95\%) \\
    Pascal & 10--10,000 (7) & 100 & 300 & 1,000 & 601.5 & 2.38 (96\%) \\
    Tarai & 2--300 (7) & 10 & 20 & 20 & 464.8 & 2.43 (96\%) \\
    \bottomrule
  \end{tabular*}
\end{table}

\paragraph{Execution time.}
Figure~\ref{fig:rq2-family-scalability} compares solving times as $n$
increases. For both Hanoi encodings and Chinese Rings, Full takes less
time at the largest parameter solved by Without DPU and continues to
solve larger tested parameters. In Ackermann and Pascal, Full takes
longer at the largest parameter solved by both configurations but
extends the solved range. Tarai shows similar times and the same solved
range with and without DPU.
The shared time axis breaks above 1200\,s. Markers on the 1800\,s budget
line identify the first unsolved $n$, with all larger tested $n$ also
unsolved. These markers denote outcomes, not measured times. Pascal and Tarai use
logarithmic $n$ axes.

\begin{figure}[htbp]
  \centering
  \includegraphics[width=0.94\linewidth]{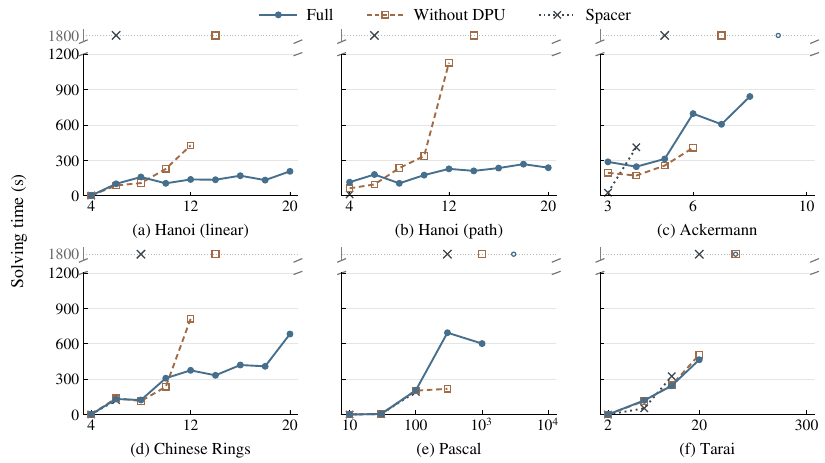}
  \caption{Scalability across the six CHC families.}
  \label{fig:rq2-family-scalability}
  \Description{Six panels in two rows show solving time against parameter n
    for Hanoi linear, Hanoi path, Ackermann, Chinese Rings, Pascal, and Tarai.
    All panels share a linear runtime scale from 0 to 1200 seconds, with an
    explicit axis break before a separate 1800-second budget line. Pascal and Tarai use
    logarithmic parameter axes. Blue circles represent Full, ochre squares
    represent Without DPU, and gray crosses represent Spacer. Curves connect
    solved runs. A marker on the dotted budget line identifies the first
    unsolved tested parameter for each configuration; all larger tested
    parameters are also unsolved. Unsolved includes unresolved runs, timeouts,
    and errors. These markers denote status, not measured runtimes.
    Nested outlines distinguish coincident boundaries.
    Each panel spans its family's full tested range.
    Full reaches larger parameters than Without DPU in five families.
    Full has lower time at the largest shared solved parameter in the two
    Hanoi encodings and Chinese Rings, and higher time there in Ackermann
    and Pascal. Both configurations reach n equal to 20 on Tarai with
    similar solving times.}
\end{figure}

\paragraph{Derivation compression.}
Following the refutation-length comparison in ADCL~\cite{SAS-2023-ADCL},
Table~\ref{tab:rq2-compression} compares compact and expanded proof sizes
for the largest solved instance in each nonlinear family. The compact
count sums the source steps in each distinct evidence entry used by the
final proof. The expanded count includes every occurrence of an original
initialization or transition clause after expanding all DPU and RF
references. Both proof counts exclude the final query. The last column
gives the largest expanded span of a DPU used by the final proof.

For Hanoi (path), 82 source uses in the compact proof represent more than
2 million original clause uses. For Chinese Rings, the corresponding
counts are 213 and about 1.4 million. In both cases, the largest used DPU
represents almost the entire expanded derivation, allowing a long
derivation to be reused through a single certified relation.
Tarai's proof contains no DPU. It relies on RF reuse to represent
$5.60\times 10^{14}$ original clause uses with 3,478 compact source uses.

These results show that DPUs can compress long derivations, while Tarai
demonstrates that RF reuse can also achieve substantial compression
without DPUs.

\begin{table}[H]
  \caption{Derivation compression on selected nonlinear instances.}
  \label{tab:rq2-compression}
  \centering
  \small
  \setlength{\tabcolsep}{4pt}
  \renewcommand{\arraystretch}{1.08}
  \begin{tabular*}{\linewidth}{@{\extracolsep{\fill}}lrrrrr@{}}
    \toprule
    Family & $n$ & \shortstack[r]{Used\\DPUs}
      & \shortstack[r]{Compact proof\\source uses}
      & \shortstack[r]{Expanded original\\clause uses}
      & \shortstack[r]{Largest used DPU\\expanded span} \\
    \midrule
    Hanoi (path) & 20 & 20 & 82 & 2,097,153 & 2,097,150 \\
    Ackermann & 8 & 122 & 3,922 & 2,785,999 & 40,630 \\
    Chinese Rings & 20 & 83 & 213 & 1,398,103 & 1,398,100 \\
    Pascal & 1,000 & 282 & 4,275 & 998,999 & 61,024 \\
    Tarai & 20 & 0 & 3,478 & $5.60\times 10^{14}$ & \textemdash \\
    \bottomrule
  \end{tabular*}
\end{table}

%% file: sections/6.3-Ablation-of-DPU.tex
\subsection{RQ3: Ablation of DPU}
\label{sec:eval-dpu}
We compare the full framework with a variant that disables DPU
construction and reuse. The variant retains RFs, invariant search,
Houdini, and Spacer. It can still guide Spacer with RFs derived through
multiple original rules. Both configurations use the same CHC-COMP
instances and parameterized families under the same resource limits.

\paragraph{CHC-COMP hard instances.}
Enabling DPU increases the number of solved Spacer-hard instances from
202 (78\%) to 228 (88\%) on the 260-instance set from RQ1
(Table~\ref{tab:rq3-chc-comp}). This raises coverage by 10 percentage
points. Most of the gain comes from unsafe results, which increase from
22 to 42. On LIA-Lin, the unsafe count grows from 17 to 38.
This gain indicates that DPUs help establish query reachability on
unsafe LIA-Lin instances. On nonlinear LIA, Full certifies four unsafe
instances and Without DPU certifies five.
Safe results increase from 180 to 186 across both tracks.
As these are single-run results, this small difference may reflect
variability in LLM-guided search.

\begin{table}[H]
  \caption{DPU ablation on Spacer-hard CHC-COMP instances.}
  \label{tab:rq3-chc-comp}
  \centering
  \small
  \setlength{\tabcolsep}{4pt}
  \renewcommand{\arraystretch}{1.08}
  \begin{tabular*}{\linewidth}{@{\extracolsep{\fill}}llrrl@{}}
    \toprule
    Track & Configuration & Safe & Unsafe & \multicolumn{1}{c@{}}{Solved / Cases} \\
    \midrule
    LIA-Lin & Full & 148 & 38 & \evalCount{186/202} (92\%) \\
            & Without DPU & 146 & 17 & \evalCount{163/202} (81\%) \\
    \midrule
    LIA & Full & 38 & 4 & \evalCount{42/58} (72\%) \\
        & Without DPU & 34 & 5 & \evalCount{39/58} (67\%) \\
    \midrule
    Total & Full & 186 & 42 & \evalCount{228/260} (88\%) \\
          & Without DPU & 180 & 22 & \evalCount{202/260} (78\%) \\
    \bottomrule
  \end{tabular*}
\end{table}

On LIA-Lin, Full certifies 21 unsafe instances that remain unresolved
without DPU (Figure~\ref{fig:rq3-runtime-scatter}). Among the 17 unsafe
instances solved by both configurations, Full is slower on 10 and faster
on 7.
In the figure, $U$ denotes unresolved runs. Instances unresolved by
both configurations are omitted.

\begin{figure}[htbp]
  \centering
  \includegraphics[width=0.82\linewidth]{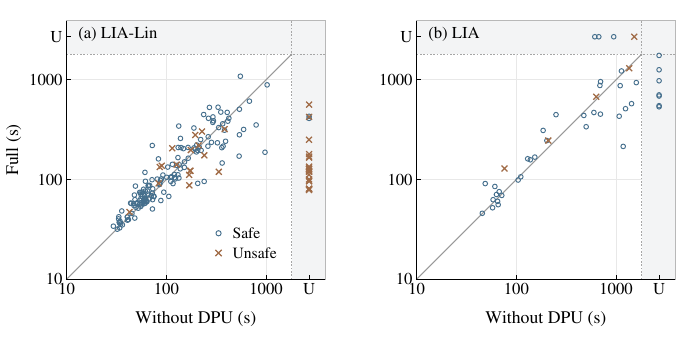}
  \caption{Runtime comparison with and without DPU.}
  \label{fig:rq3-runtime-scatter}
  \Description{Two scatter plots compare Without DPU on the horizontal axis
    with Full on the vertical axis, for LIA-Lin and LIA. Blue circles denote safe
    outcomes and ochre crosses denote unsafe outcomes. Points below the
    diagonal indicate faster Full runs. Shaded margins marked U are categorical
    unresolved outcomes, not runtimes. Points in the right margin are solved
    only by Full. Points in the top margin are solved only without DPU.
    The right margin for LIA-Lin contains 21 unsafe instances solved only by Full.}
\end{figure}

\paragraph{Depth-and-Branching CHC Suite.}
Enabling DPU also increases the number of solved family instances from
27 to 42 out of 49. On the five nonlinear families, coverage increases
from 22/40 to 33/40. Four of these families reach larger parameters with
DPU (Table~\ref{tab:rq2-families}). The maximum solved parameter rises
from 12 to 20 for Hanoi (path) and Chinese Rings, from 6 to 8 for
Ackermann, and from 300 to 1000 for Pascal. Both configurations reach
$n=20$ on Tarai.

%% file: sections/6.4-Limitations.tex
\subsection{Limitations}
\label{sec:eval-limitations}

Our CHC-COMP evaluation is limited to integer-arithmetic instances left
unresolved by the selected Spacer configuration. The parameterized
families probe specific patterns of deep and branching derivations.
These benchmark sets support conclusions about hard-instance solving
and derivation scalability, but not an overall ranking on the full
CHC-COMP suite.

LLM-guided search can vary across runs and model revisions. We report
results from the recorded campaigns without repeated-run averages.
Although all configurations share benchmark sets and resource limits,
repeated runs are needed to assess the stability of the observed
differences in solving coverage and runtime.

%% file: sections/7-Related-Work.tex
\section{Related Work}
\label{sec:related-work}

\noindent\textbf{Traditional CHC Solving.}
Recent advances in CHC solving have led to mature solvers based on a diverse
set of techniques~\cite{CAV-2022-Spacer-Invited,CAV-2023-Golem}.
Spacer extends IC3 with reachable-state under-approximations
and model-based projection for lazy quantifier 
elimination~\cite{VMCAI-2011-IC3,CAV-2014-Spacer,CAV-2022-Spacer-Invited}.
HornCraft's reach facts play a
similar role but record only individual SMT models.
Subsequent work further develops Spacer through global guidance and inductive
reformulations~\cite{CAV-2020-GSpacer,PLDI-2024-MuCyc}.
Eldarica and Ultimate TreeAutomizer both use tree
interpolation~\cite{FMCAD-2018-Eldarica,
EPTCS-2019-Ultimate-TreeAutomizer}. Theta translates CHCs
to control-flow automata for model checking~\cite{ETPCS-2025-Theta}.
PCSat combines CEGIS with probabilistic inference~\cite{AAAI-2020-PCSat}.
Horn-ICE learns invariants and contracts from Horn implication
counterexamples~\cite{OOPSLA-2018-Horn-ICE}.
Beyond solving, PyCHC validates CHC solver results using independent
SMT solvers and proof checkers~\cite{CAV-2026-PyCHC}.

Related approaches summarize multi-step behavior through acceleration,
transition powers, or cycle summaries~\cite{SAS-2023-ADCL,TACAS-2022-TPA,IFM-2025-CSRA}.
ADCL, implemented in LoAT~\cite{CADE-2025-LoAT},
learns accelerated clauses that capture or
under-approximate the closure of repeated recursive
derivations~\cite{SAS-2023-ADCL}.
TPA, implemented in Golem, maintains
over-approximations of bounded transition
powers~\cite{TACAS-2022-TPA,FMCAD-2022-Split-Split-TPA,CAV-2023-Golem}.
ADCL and TPA have shown strong empirical
performance on deep linear instances. DPUs resemble ADCL 
accelerations in under-approximating multi-step behavior and
TPA in composing shorter relations into longer ones. DPUs additionally
support nonlinear CHCs.

\par\noindent\textbf{LLM-Assisted Formal Verification.}
For software verification, LLMs have been used to predict program
invariants~\cite{ICML-2023-LLM-Invariants}, generate loop invariants for
symbolic validation~\cite{EMNLP-2023-iRank,FMCAD-2024-Loopy,
ASE-2024-LaM4Inv,ASE-2024-ESBMC-AI,ISSTA-2025-Clause2Inv,
arXiv-2026-Quokka,TOPLAS-2026-LORIS,arXiv-2026-LimICE}, and propose auxiliary
assertions for symbolic checking~\cite{ICLR-2024-LEMUR}.
The guess-and-check paradigm has also been studied in model checking through
invariant and lemma generation~\cite{arXiv-2026-CILL,
IJCAI-2026-FindIndStrength,arXiv-2026-IC3Syn} and verifier-checked semantic guidance for
word-level PDR~\cite{arXiv-2026-LLM4PDR}, as well as in automated inductive
reasoning~\cite{FM-2026-LLM4Ind}. LLM agents have also been used to construct checked program
proofs~\cite{FSE-2026-AutoRocq} and to repair failed proofs using verifier
feedback~\cite{arXiv-2026-ExVerus,arXiv-2026-Schwarz}.
Following this paradigm, HornCraft asks the LLM to propose reusable
under-approximations of multi-step CHC derivations together with symbolic
witnesses that enable validation by the solver without introducing additional
quantifiers.

%% file: sections/8-Conclusion.tex
\section{Conclusion}
\label{sec:conclusion}

We present HornCraft, an LLM-assisted bi-approximate framework
for CHC solving that combines DPU-based compression of reachability
derivations with invariant generation for unreachability.
Symbolic witnesses enable SMT certification of DPU relations
without introducing additional quantifiers, and certified DPUs
can be soundly composed and reused in linear and branching
derivations. HornCraft solved 228 of 260 CHC-COMP instances
unresolved by a 30-minute Spacer run. Compared with the variant
without DPUs, the full framework certified 42 rather
than 22 unsafe instances. On the Depth-and-Branching CHC Suite,
it also solved larger instances in four nonlinear families.
DPUs thus provide a formally checked interface for using
LLM-proposed symbolic relations as reusable proof steps
for discovering deep and branching counterexamples.

%% file: references.bib
@inproceedings{CAV-2016-JayHorn,
  author       = {Temesghen Kahsai and
                  Philipp R{\"{u}}mmer and
                  Huascar Sanchez and
                  Martin Sch{\"{a}}f},
  editor       = {Swarat Chaudhuri and
                  Azadeh Farzan},
  title        = {JayHorn: {A} Framework for Verifying Java programs},
  booktitle    = {Computer Aided Verification - 28th International Conference, {CAV}
                  2016, Toronto, ON, Canada, July 17-23, 2016, Proceedings, Part {I}},
  series       = {Lecture Notes in Computer Science},
  volume       = {9779},
  pages        = {352--358},
  publisher    = {Springer},
  year         = {2016},
  url          = {https://doi.org/10.1007/978-3-319-41528-4\_19},
  doi          = {10.1007/978-3-319-41528-4\_19},
  bibsource    = {dblp computer science bibliography, https://dblp.org}
}

@inproceedings{EuroSP-2016-HornDroid,
  author       = {Stefano Calzavara and
                  Ilya Grishchenko and
                  Matteo Maffei},
  title        = {HornDroid: Practical and Sound Static Analysis of Android Applications
                  by {SMT} Solving},
  booktitle    = {{IEEE} European Symposium on Security and Privacy, EuroS{\&}P
                  2016, Saarbr{\"{u}}cken, Germany, March 21-24, 2016},
  pages        = {47--62},
  publisher    = {{IEEE}},
  year         = {2016},
  url          = {https://doi.org/10.1109/EuroSP.2016.16},
  doi          = {10.1109/EUROSP.2016.16},
  bibsource    = {dblp computer science bibliography, https://dblp.org}
}

@article{TOPLAS-2021-RustHorn,
  author       = {Yusuke Matsushita and
                  Takeshi Tsukada and
                  Naoki Kobayashi},
  title        = {RustHorn: CHC-based Verification for Rust Programs},
  journal      = {{ACM} Trans. Program. Lang. Syst.},
  volume       = {43},
  number       = {4},
  pages        = {15:1--15:54},
  year         = {2021},
  url          = {https://doi.org/10.1145/3462205},
  doi          = {10.1145/3462205},
  bibsource    = {dblp computer science bibliography, https://dblp.org}
}

@inproceedings{VMCAI-2022-SmartACE,
  author       = {Scott Wesley and
                  Maria Christakis and
                  Jorge A. Navas and
                  Richard J. Trefler and
                  Valentin W{\"{u}}stholz and
                  Arie Gurfinkel},
  editor       = {Bernd Finkbeiner and
                  Thomas Wies},
  title        = {Verifying Solidity Smart Contracts via Communication Abstraction in
                  SmartACE},
  booktitle    = {Verification, Model Checking, and Abstract Interpretation - 23rd International
                  Conference, {VMCAI} 2022, Philadelphia, PA, USA, January 16-18, 2022,
                  Proceedings},
  series       = {Lecture Notes in Computer Science},
  volume       = {13182},
  pages        = {425--449},
  publisher    = {Springer},
  year         = {2022},
  url          = {https://doi.org/10.1007/978-3-030-94583-1\_21},
  doi          = {10.1007/978-3-030-94583-1\_21},
  bibsource    = {dblp computer science bibliography, https://dblp.org}
}

@inproceedings{CAV-2022-SolCMC,
  author       = {Leonardo Alt and
                  Martin Blicha and
                  Antti E. J. Hyv{\"{a}}rinen and
                  Natasha Sharygina},
  editor       = {Sharon Shoham and
                  Yakir Vizel},
  title        = {SolCMC: Solidity Compiler's Model Checker},
  booktitle    = {Computer Aided Verification - 34th International Conference, {CAV}
                  2022, Haifa, Israel, August 7-10, 2022, Proceedings, Part {I}},
  series       = {Lecture Notes in Computer Science},
  volume       = {13371},
  pages        = {325--338},
  publisher    = {Springer},
  year         = {2022},
  url          = {https://doi.org/10.1007/978-3-031-13185-1\_16},
  doi          = {10.1007/978-3-031-13185-1\_16},
  bibsource    = {dblp computer science bibliography, https://dblp.org}
}

@inproceedings{VMCAI-2022-Korn,
  author       = {Gidon Ernst},
  editor       = {Bernd Finkbeiner and
                  Thomas Wies},
  title        = {Loop Verification with Invariants and Contracts},
  booktitle    = {Verification, Model Checking, and Abstract Interpretation - 23rd International
                  Conference, {VMCAI} 2022, Philadelphia, PA, USA, January 16-18, 2022,
                  Proceedings},
  series       = {Lecture Notes in Computer Science},
  volume       = {13182},
  pages        = {69--92},
  publisher    = {Springer},
  year         = {2022},
  url          = {https://doi.org/10.1007/978-3-030-94583-1\_4},
  doi          = {10.1007/978-3-030-94583-1\_4},
  bibsource    = {dblp computer science bibliography, https://dblp.org}
}

@inproceedings{EPTCS-2019-Ultimate-TreeAutomizer,
  author       = {Daniel Dietsch and
                  Matthias Heizmann and
                  Jochen Hoenicke and
                  Alexander Nutz and
                  Andreas Podelski},
  editor       = {Emanuele De Angelis and
                  Grigory Fedyukovich and
                  Nikos Tzevelekos and
                  Mattias Ulbrich},
  title        = {Ultimate TreeAutomizer {(CHC-COMP} Tool Description)},
  booktitle    = {Proceedings of the Sixth Workshop on Horn Clauses for Verification
                  and Synthesis and Third Workshop on Program Equivalence and Relational
                  Reasoning, HCVS/PERR@ETAPS 2019, Prague, Czech Republic, 6-7th April
                  2019},
  series       = {{EPTCS}},
  volume       = {296},
  pages        = {42--47},
  year         = {2019},
  url          = {https://doi.org/10.4204/EPTCS.296.7},
  doi          = {10.4204/EPTCS.296.7},
  bibsource    = {dblp computer science bibliography, https://dblp.org}
}

@inproceedings{ETPCS-2025-Theta,
  author       = {Levente Bajczi and
                  Mil{\'{a}}n Mondok and
                  Vince Moln{\'{a}}r},
  editor       = {Emanuele De Angelis and
                  Florian Frohn},
  title        = {Theta as a Horn Solver},
  booktitle    = {Proceedings of the 12th Workshop on Horn Clauses for Verification
                  and Synthesis, HCVS@CAV 2025, Zagreb, Croatia, 22 July 2025},
  series       = {{EPTCS}},
  volume       = {434},
  pages        = {27--39},
  year         = {2025},
  month        = oct,
  url          = {https://doi.org/10.4204/EPTCS.434.5},
  doi          = {10.4204/EPTCS.434.5},
  bibsource    = {dblp computer science bibliography, https://dblp.org}
}

@inproceedings{AAAI-2020-PCSat,
  author       = {Yuki Satake and
                  Hiroshi Unno and
                  Hinata Yanagi},
  title        = {Probabilistic Inference for Predicate Constraint Satisfaction},
  booktitle    = {The Thirty-Fourth {AAAI} Conference on Artificial Intelligence, {AAAI}
                  2020, The Thirty-Second Innovative Applications of Artificial Intelligence
                  Conference, {IAAI} 2020, The Tenth {AAAI} Symposium on Educational
                  Advances in Artificial Intelligence, {EAAI} 2020, New York, NY, USA,
                  February 7-12, 2020},
  pages        = {1644--1651},
  publisher    = {{AAAI} Press},
  year         = {2020},
  url          = {https://doi.org/10.1609/aaai.v34i02.5526},
  doi          = {10.1609/AAAI.V34I02.5526},
  bibsource    = {dblp computer science bibliography, https://dblp.org}
}

@inproceedings{TACAS-2008-Z3,
  author       = {Leonardo Mendon{\c{c}}a de Moura and
                  Nikolaj S. Bj{\o}rner},
  editor       = {C. R. Ramakrishnan and
                  Jakob Rehof},
  title        = {{Z3:} An Efficient {SMT} Solver},
  booktitle    = {Tools and Algorithms for the Construction and Analysis of Systems,
                  14th International Conference, {TACAS} 2008, Held as Part of the Joint
                  European Conferences on Theory and Practice of Software, {ETAPS} 2008,
                  Budapest, Hungary, March 29-April 6, 2008. Proceedings},
  series       = {Lecture Notes in Computer Science},
  volume       = {4963},
  pages        = {337--340},
  publisher    = {Springer},
  year         = {2008},
  url          = {https://doi.org/10.1007/978-3-540-78800-3\_24},
  doi          = {10.1007/978-3-540-78800-3\_24},
  bibsource    = {dblp computer science bibliography, https://dblp.org}
}

@inproceedings{CAV-2015-Seahorn,
  author       = {Arie Gurfinkel and
                  Temesghen Kahsai and
                  Anvesh Komuravelli and
                  Jorge A. Navas},
  editor       = {Daniel Kroening and
                  Corina S. Pasareanu},
  title        = {The SeaHorn Verification Framework},
  booktitle    = {Computer Aided Verification - 27th International Conference, {CAV}
                  2015, San Francisco, CA, USA, July 18-24, 2015, Proceedings, Part
                  {I}},
  series       = {Lecture Notes in Computer Science},
  volume       = {9206},
  pages        = {343--361},
  publisher    = {Springer},
  year         = {2015},
  url          = {https://doi.org/10.1007/978-3-319-21690-4\_20},
  doi          = {10.1007/978-3-319-21690-4\_20},
  bibsource    = {dblp computer science bibliography, https://dblp.org}
}

@inproceedings{CAV-2026-DrvHorn,
  author       = {Joe Hattori and
                  Naoki Kobayashi and
                  Ken Sakayori},
  editor       = {Eva Darulova and
                  Anthony W. Lin and
                  Philipp R{\"{u}}mmer},
  title        = {Automatic Detection of Reference Counting Bugs in Linux Kernel Drivers},
  booktitle    = {Computer Aided Verification - 38th International Conference, {CAV}
                  2026, Lisbon, Portugal, July 26-29, 2026, Proceedings, Part {III}},
  series       = {Lecture Notes in Computer Science},
  volume       = {16684},
  pages        = {191--203},
  publisher    = {Springer},
  year         = {2026},
  url          = {https://doi.org/10.1007/978-3-032-32537-2\_9},
  doi          = {10.1007/978-3-032-32537-2\_9},
  bibsource    = {dblp computer science bibliography, https://dblp.org}
}

@inproceedings{CAV-2023-Golem,
  author       = {Martin Blicha and
                  Konstantin Britikov and
                  Natasha Sharygina},
  editor       = {Constantin Enea and
                  Akash Lal},
  title        = {The Golem Horn Solver},
  booktitle    = {Computer Aided Verification - 35th International Conference, {CAV}
                  2023, Paris, France, July 17-22, 2023, Proceedings, Part {II}},
  series       = {Lecture Notes in Computer Science},
  volume       = {13965},
  pages        = {209--223},
  publisher    = {Springer},
  year         = {2023},
  url          = {https://doi.org/10.1007/978-3-031-37703-7\_10},
  doi          = {10.1007/978-3-031-37703-7\_10},
  bibsource    = {dblp computer science bibliography, https://dblp.org}
}

@inproceedings{FMCAD-2018-Eldarica,
  author       = {Hossein Hojjat and
                  Philipp R{\"{u}}mmer},
  editor       = {Nikolaj S. Bj{\o}rner and
                  Arie Gurfinkel},
  title        = {The {ELDARICA} Horn Solver},
  booktitle    = {2018 Formal Methods in Computer Aided Design, {FMCAD} 2018, Austin,
                  TX, USA, October 30 - November 2, 2018},
  pages        = {1--7},
  publisher    = {{IEEE}},
  year         = {2018},
  url          = {https://doi.org/10.23919/FMCAD.2018.8603013},
  doi          = {10.23919/FMCAD.2018.8603013},
  bibsource    = {dblp computer science bibliography, https://dblp.org}
}

@inproceedings{CAV-2026-PyCHC,
  author       = {Anna Becchi and
                  Martin Blicha and
                  Rodrigo Otoni and
                  Natasha Sharygina},
  editor       = {Eva Darulova and
                  Anthony W. Lin and
                  Philipp R{\"{u}}mmer},
  title        = {PyCHC: {A} Framework for Certified Horn Solving and CHC-Based Design},
  booktitle    = {Computer Aided Verification - 38th International Conference, {CAV}
                  2026, Lisbon, Portugal, July 26-29, 2026, Proceedings, Part {III}},
  series       = {Lecture Notes in Computer Science},
  volume       = {16684},
  pages        = {261--275},
  publisher    = {Springer},
  year         = {2026},
  url          = {https://doi.org/10.1007/978-3-032-32537-2\_13},
  doi          = {10.1007/978-3-032-32537-2\_13},
  bibsource    = {dblp computer science bibliography, https://dblp.org}
}

@dataset{CHC-COMP-2026-Results,
  author       = {Bajczi, Levente and
                  Britikov, Konstantin},
  title        = {CHC-COMP'26 results and competition infrastructure},
  month        = may,
  year         = 2026,
  publisher    = {Zenodo},
  doi          = {10.5281/zenodo.20413019},
  url          = {https://doi.org/10.5281/zenodo.20413019},
}

@inproceedings{VMCAI-2011-IC3,
  author       = {Aaron R. Bradley},
  editor       = {Ranjit Jhala and
                  David A. Schmidt},
  title        = {SAT-Based Model Checking without Unrolling},
  booktitle    = {Verification, Model Checking, and Abstract Interpretation - 12th International
                  Conference, {VMCAI} 2011, Austin, TX, USA, January 23-25, 2011. Proceedings},
  series       = {Lecture Notes in Computer Science},
  volume       = {6538},
  pages        = {70--87},
  publisher    = {Springer},
  year         = {2011},
  url          = {https://doi.org/10.1007/978-3-642-18275-4\_7},
  doi          = {10.1007/978-3-642-18275-4\_7},
  bibsource    = {dblp computer science bibliography, https://dblp.org}
}

@inproceedings{SAT-2012-GPDR,
  author       = {Krystof Hoder and
                  Nikolaj S. Bj{\o}rner},
  editor       = {Alessandro Cimatti and
                  Roberto Sebastiani},
  title        = {Generalized Property Directed Reachability},
  booktitle    = {Theory and Applications of Satisfiability Testing - {SAT} 2012 - 15th
                  International Conference, Trento, Italy, June 17-20, 2012. Proceedings},
  series       = {Lecture Notes in Computer Science},
  volume       = {7317},
  pages        = {157--171},
  publisher    = {Springer},
  year         = {2012},
  url          = {https://doi.org/10.1007/978-3-642-31612-8\_13},
  doi          = {10.1007/978-3-642-31612-8\_13},
  bibsource    = {dblp computer science bibliography, https://dblp.org}
}

@inproceedings{CAV-2014-Spacer,
  author       = {Anvesh Komuravelli and
                  Arie Gurfinkel and
                  Sagar Chaki},
  editor       = {Armin Biere and
                  Roderick Bloem},
  title        = {SMT-Based Model Checking for Recursive Programs},
  booktitle    = {Computer Aided Verification - 26th International Conference, {CAV}
                  2014, Held as Part of the Vienna Summer of Logic, {VSL} 2014, Vienna,
                  Austria, July 18-22, 2014. Proceedings},
  series       = {Lecture Notes in Computer Science},
  volume       = {8559},
  pages        = {17--34},
  publisher    = {Springer},
  year         = {2014},
  url          = {https://doi.org/10.1007/978-3-319-08867-9\_2},
  doi          = {10.1007/978-3-319-08867-9\_2},
  bibsource    = {dblp computer science bibliography, https://dblp.org}
}

@inproceedings{CAV-2022-Spacer-Invited,
  author       = {Arie Gurfinkel},
  editor       = {Sharon Shoham and
                  Yakir Vizel},
  title        = {Program Verification with Constrained Horn Clauses (Invited Paper)},
  booktitle    = {Computer Aided Verification - 34th International Conference, {CAV}
                  2022, Haifa, Israel, August 7-10, 2022, Proceedings, Part {I}},
  series       = {Lecture Notes in Computer Science},
  volume       = {13371},
  pages        = {19--29},
  publisher    = {Springer},
  year         = {2022},
  url          = {https://doi.org/10.1007/978-3-031-13185-1\_2},
  doi          = {10.1007/978-3-031-13185-1\_2},
  bibsource    = {dblp computer science bibliography, https://dblp.org}
}

@inproceedings{CAV-2020-GSpacer,
  author       = {Hari Govind Vediramana Krishnan and
                  Yuting Chen and
                  Sharon Shoham and
                  Arie Gurfinkel},
  editor       = {Shuvendu K. Lahiri and
                  Chao Wang},
  title        = {Global Guidance for Local Generalization in Model Checking},
  booktitle    = {Computer Aided Verification - 32nd International Conference, {CAV}
                  2020, Los Angeles, CA, USA, July 21-24, 2020, Proceedings, Part {II}},
  series       = {Lecture Notes in Computer Science},
  volume       = {12225},
  pages        = {101--125},
  publisher    = {Springer},
  year         = {2020},
  url          = {https://doi.org/10.1007/978-3-030-53291-8\_7},
  doi          = {10.1007/978-3-030-53291-8\_7},
  bibsource    = {dblp computer science bibliography, https://dblp.org}
}

@article{PLDI-2024-MuCyc,
  author       = {Takeshi Tsukada and
                  Hiroshi Unno},
  title        = {Inductive Approach to Spacer},
  journal      = {Proc. {ACM} Program. Lang.},
  volume       = {8},
  number       = {{PLDI}},
  pages        = {1979--2002},
  year         = {2024},
  url          = {https://doi.org/10.1145/3656457},
  doi          = {10.1145/3656457},
  bibsource    = {dblp computer science bibliography, https://dblp.org}
}

@inproceedings{FMCAD-2022-Split-Split-TPA,
  author       = {Martin Blicha and
                  Grigory Fedyukovich and
                  Antti E. J. Hyv{\"{a}}rinen and
                  Natasha Sharygina},
  editor       = {Alberto Griggio and
                  Neha Rungta},
  title        = {Split Transition Power Abstraction for Unbounded Safety},
  booktitle    = {22nd Formal Methods in Computer-Aided Design, {FMCAD} 2022, Trento,
                  Italy, October 17-21, 2022},
  pages        = {349--358},
  publisher    = {{IEEE}},
  year         = {2022},
  url          = {https://doi.org/10.34727/2022/isbn.978-3-85448-053-2\_42},
  doi          = {10.34727/2022/ISBN.978-3-85448-053-2\_42},
  bibsource    = {dblp computer science bibliography, https://dblp.org}
}

@inproceedings{TACAS-2022-TPA,
  author       = {Martin Blicha and
                  Grigory Fedyukovich and
                  Antti E. J. Hyv{\"{a}}rinen and
                  Natasha Sharygina},
  editor       = {Dana Fisman and
                  Grigore Rosu},
  title        = {Transition Power Abstractions for Deep Counterexample Detection},
  booktitle    = {Tools and Algorithms for the Construction and Analysis of Systems
                  - 28th International Conference, {TACAS} 2022, Held as Part of the
                  European Joint Conferences on Theory and Practice of Software, {ETAPS}
                  2022, Munich, Germany, April 2-7, 2022, Proceedings, Part {I}},
  series       = {Lecture Notes in Computer Science},
  volume       = {13243},
  pages        = {524--542},
  publisher    = {Springer},
  year         = {2022},
  url          = {https://doi.org/10.1007/978-3-030-99524-9\_29},
  doi          = {10.1007/978-3-030-99524-9\_29},
  bibsource    = {dblp computer science bibliography, https://dblp.org}
}

@inproceedings{IFM-2025-CSRA,
  author       = {Konstantin Britikov and
                  Grigory Fedyukovich and
                  Natasha Sharygina},
  editor       = {Ferruccio Damiani and
                  Marie Farrell},
  title        = {CHC-Based Reachability Analysis via Cycle Summarization},
  booktitle    = {Integrated Formal Methods - 20th International Conference, iFM 2025,
                  Paris, France, November 19-21, 2025, Proceedings},
  series       = {Lecture Notes in Computer Science},
  volume       = {16194},
  pages        = {205--225},
  publisher    = {Springer},
  year         = {2025},
  url          = {https://doi.org/10.1007/978-3-032-10794-7\_11},
  doi          = {10.1007/978-3-032-10794-7\_11},
  bibsource    = {dblp computer science bibliography, https://dblp.org}
}

@inproceedings{SAS-2023-ADCL,
  author       = {Florian Frohn and
                  J{\"{u}}rgen Giesl},
  editor       = {Manuel V. Hermenegildo and
                  Jos{\'{e}} F. Morales},
  title        = {{ADCL:} Acceleration Driven Clause Learning for Constrained Horn Clauses},
  booktitle    = {Static Analysis - 30th International Symposium, {SAS} 2023, Cascais,
                  Portugal, October 22-24, 2023, Proceedings},
  series       = {Lecture Notes in Computer Science},
  volume       = {14284},
  pages        = {259--285},
  publisher    = {Springer},
  year         = {2023},
  url          = {https://doi.org/10.1007/978-3-031-44245-2\_13},
  doi          = {10.1007/978-3-031-44245-2\_13},
  bibsource    = {dblp computer science bibliography, https://dblp.org}
}

@inproceedings{CADE-2025-LoAT,
  author       = {Florian Frohn and
                  J{\"{u}}rgen Giesl},
  editor       = {Clark W. Barrett and
                  Uwe Waldmann},
  title        = {Infinite State Model Checking by Learning Transitive Relations},
  booktitle    = {Automated Deduction - {CADE} 30 - 30th International Conference on
                  Automated Deduction, Stuttgart, Germany, July 28-31, 2025, Proceedings},
  series       = {Lecture Notes in Computer Science},
  volume       = {15943},
  pages        = {468--489},
  publisher    = {Springer},
  year         = {2025},
  url          = {https://doi.org/10.1007/978-3-031-99984-0\_25},
  doi          = {10.1007/978-3-031-99984-0\_25},
  bibsource    = {dblp computer science bibliography, https://dblp.org}
}

@article{OOPSLA-2018-Horn-ICE,
  author       = {P. Ezudheen and
                  Daniel Neider and
                  Deepak D'Souza and
                  Pranav Garg and
                  P. Madhusudan},
  title        = {Horn-ICE learning for synthesizing invariants and contracts},
  journal      = {Proc. {ACM} Program. Lang.},
  volume       = {2},
  number       = {{OOPSLA}},
  pages        = {131:1--131:25},
  year         = {2018},
  url          = {https://doi.org/10.1145/3276501},
  doi          = {10.1145/3276501},
  bibsource    = {dblp computer science bibliography, https://dblp.org}
}

@inproceedings{FME-2001-Houdini,
  author       = {Cormac Flanagan and
                  K. Rustan M. Leino},
  editor       = {Jos{\'{e}} Nuno Oliveira and
                  Pamela Zave},
  title        = {Houdini, an Annotation Assistant for ESC/Java},
  booktitle    = {{FME} 2001: Formal Methods for Increasing Software Productivity, International
                  Symposium of Formal Methods Europe, Berlin, Germany, March 12-16,
                  2001, Proceedings},
  series       = {Lecture Notes in Computer Science},
  volume       = {2021},
  pages        = {500--517},
  publisher    = {Springer},
  year         = {2001},
  url          = {https://doi.org/10.1007/3-540-45251-6\_29},
  doi          = {10.1007/3-540-45251-6\_29},
  bibsource    = {dblp computer science bibliography, https://dblp.org}
}

@article{arXiv-2026-Schwarz,
  author       = {Jingyu Ke and
                  Ling{-}I Wu and
                  Guoqiang Li},
  title        = {Schwarz: Solver-Aware Agentic Program Verification},
  journal      = {CoRR},
  volume       = {abs/2608.30803},
  year         = {2026},
  url          = {https://doi.org/10.48550/arXiv.2608.30803},
  doi          = {10.48550/ARXIV.2608.30803},
  eprinttype   = {arXiv},
  eprint       = {2608.30803},
  bibsource    = {dblp computer science bibliography, https://dblp.org}
}

@article{arXiv-2026-AProver,
  author       = {Youcheng Sun and
                  Jiawen Liu and
                  Daniel Kroening and
                  Jason Xue},
  title        = {Agentic Model Checking},
  journal      = {CoRR},
  volume       = {abs/2605.21434},
  year         = {2026},
  url          = {https://doi.org/10.48550/arXiv.2605.21434},
  doi          = {10.48550/ARXIV.2605.21434},
  eprinttype   = {arXiv},
  eprint       = {2605.21434},
  bibsource    = {dblp computer science bibliography, https://dblp.org}
}

@article{FSE-2026-AutoRocq,
  author       = {Haoxin Tu and
                  Huan Zhao and
                  Yahui Song and
                  Mehtab Zafar and
                  Ruijie Meng and
                  Abhik Roychoudhury},
  title        = {Agentic Verification of Software Systems},
  journal      = {Proc. {ACM} Softw. Eng.},
  volume       = {3},
  number       = {{FSE}},
  pages        = {3558--3581},
  year         = {2026},
  url          = {https://doi.org/10.1145/3808164},
  doi          = {10.1145/3808164},
  bibsource    = {dblp computer science bibliography, https://dblp.org}
}

@article{arXiv-2026-ExVerus,
  author       = {Jun Yang and
                  Yuechun Sun and
                  Yi Wu and
                  Rodrigo Caridad and
                  Yongwei Yuan and
                  Jianan Yao and
                  Shan Lu and
                  Kexin Pei},
  title        = {ExVerus: Verus Proof Repair via Counterexample Reasoning},
  journal      = {CoRR},
  volume       = {abs/2603.25810},
  year         = {2026},
  url          = {https://doi.org/10.48550/arXiv.2603.25810},
  doi          = {10.48550/ARXIV.2603.25810},
  eprinttype   = {arXiv},
  eprint       = {2603.25810},
  bibsource    = {dblp computer science bibliography, https://dblp.org}
}

@article{arXiv-2026-LLM4PDR,
  author       = {Guangyu Hu and
                  Mingkai Miao and
                  Zhiyuan Yan and
                  Xiaofeng Zhou and
                  Wei Zhang and
                  Hongce Zhang},
  title        = {Enhancing Word-Level Property Directed Reachability with LLM-Driven
                  Semantic Guidance},
  journal      = {CoRR},
  volume       = {abs/2609.30131},
  year         = {2026},
  url          = {https://doi.org/10.48550/arXiv.2609.30131},
  doi          = {10.48550/ARXIV.2609.30131},
  eprinttype   = {arXiv},
  eprint       = {2609.30131},
  bibsource    = {dblp computer science bibliography, https://dblp.org}
}

@inproceedings{ICML-2023-LLM-Invariants,
  author       = {Kexin Pei and
                  David Bieber and
                  Kensen Shi and
                  Charles Sutton and
                  Pengcheng Yin},
  editor       = {Andreas Krause and
                  Emma Brunskill and
                  Kyunghyun Cho and
                  Barbara Engelhardt and
                  Sivan Sabato and
                  Jonathan Scarlett},
  title        = {Can Large Language Models Reason about Program Invariants?},
  booktitle    = {International Conference on Machine Learning, {ICML} 2023, 23-29 July
                  2023, Honolulu, Hawaii, {USA}},
  series       = {Proceedings of Machine Learning Research},
  volume       = {202},
  pages        = {27496--27520},
  publisher    = {{PMLR}},
  year         = {2023},
  url          = {https://proceedings.mlr.press/v202/pei23a.html},
  bibsource    = {dblp computer science bibliography, https://dblp.org}
}

@inproceedings{FMCAD-2024-Loopy,
  author       = {Adharsh Kamath and
                  J. Nausheen Mohammed and
                  Aditya Senthilnathan and
                  Saikat Chakraborty and
                  Pantazis Deligiannis and
                  Shuvendu K. Lahiri and
                  Akash Lal and
                  Aseem Rastogi and
                  Subhajit Roy and
                  Rahul Sharma},
  editor       = {Nina Narodytska and
                  Philipp R{\"{u}}mmer},
  title        = {Leveraging LLMs for Program Verification},
  booktitle    = {Formal Methods in Computer-Aided Design, {FMCAD} 2024, Prague, Czech
                  Republic, October 15-18, 2024},
  pages        = {107--118},
  publisher    = {{IEEE}},
  year         = {2024},
  url          = {https://doi.org/10.34727/2024/isbn.978-3-85448-065-5\_16},
  doi          = {10.34727/2024/ISBN.978-3-85448-065-5\_16},
  bibsource    = {dblp computer science bibliography, https://dblp.org}
}

@inproceedings{EMNLP-2023-iRank,
  author       = {Saikat Chakraborty and
                  Shuvendu K. Lahiri and
                  Sarah Fakhoury and
                  Akash Lal and
                  Madanlal Musuvathi and
                  Aseem Rastogi and
                  Aditya Senthilnathan and
                  Rahul Sharma and
                  Nikhil Swamy},
  editor       = {Houda Bouamor and
                  Juan Pino and
                  Kalika Bali},
  title        = {Ranking LLM-Generated Loop Invariants for Program Verification},
  booktitle    = {Findings of the Association for Computational Linguistics: {EMNLP}
                  2023, Singapore, December 6-10, 2023},
  series       = {Findings of {ACL}},
  volume       = {{EMNLP} 2023},
  pages        = {9164--9175},
  publisher    = {Association for Computational Linguistics},
  year         = {2023},
  url          = {https://doi.org/10.18653/v1/2023.findings-emnlp.614},
  doi          = {10.18653/V1/2023.FINDINGS-EMNLP.614},
  bibsource    = {dblp computer science bibliography, https://dblp.org}
}

@article{ISSTA-2025-Clause2Inv,
  author       = {Weining Cao and
                  Guangyuan Wu and
                  Tangzhi Xu and
                  Yuan Yao and
                  Hengfeng Wei and
                  Taolue Chen and
                  Xiaoxing Ma},
  title        = {Clause2Inv: {A} Generate-Combine-Check Framework for Loop Invariant
                  Inference},
  journal      = {Proc. {ACM} Softw. Eng.},
  volume       = {2},
  number       = {{ISSTA}},
  pages        = {1009--1030},
  year         = {2025},
  url          = {https://doi.org/10.1145/3728920},
  doi          = {10.1145/3728920},
  bibsource    = {dblp computer science bibliography, https://dblp.org}
}

@article{arXiv-2026-Quokka,
  title={Quokka: Accelerating Program Verification with LLMs via Invariant Synthesis},
  author={Wei, Anjiang and Sun, Tianran and Suresh, Tarun and Wu, Haoze and Wang, Ke and Aiken, Alex},
  journal={arXiv preprint arXiv:2509.21629},
  year={2025}
}

@article{TOPLAS-2026-LORIS,
  author       = {Tianchi Li and
                  Zhenyu Yan and
                  Junhao Liu and
                  Peng Di and
                  Xin Zhang},
  title        = {Guiding LLM-Based Loop Invariant Synthesis via Feedback on Local Reasoning
                  Errors},
  journal      = {{ACM} Trans. Program. Lang. Syst.},
  volume       = {48},
  number       = {2},
  pages        = {8:1--8:30},
  year         = {2026},
  url          = {https://doi.org/10.1145/3806652},
  doi          = {10.1145/3806652},
  bibsource    = {dblp computer science bibliography, https://dblp.org}
}

@article{arXiv-2026-IC3Syn,
  author       = {Weining Cao and
                  Guangyuan Wu and
                  Yuan Yao and
                  Hengfeng Wei and
                  Taolue Chen and
                  Xiaoxing Ma},
  title        = {Synthesizing Inductive Invariants for Distributed Protocols via {IC3}
                  and Large Language Models},
  journal      = {CoRR},
  volume       = {abs/2605.24619},
  year         = {2026},
  url          = {https://doi.org/10.48550/arXiv.2605.24619},
  doi          = {10.48550/ARXIV.2605.24619},
  eprinttype   = {arXiv},
  eprint       = {2605.24619},
  bibsource    = {dblp computer science bibliography, https://dblp.org}
}

@inproceedings{ICLR-2024-LEMUR,
  author       = {Haoze Wu and
                  Clark W. Barrett and
                  Nina Narodytska},
  title        = {Lemur: Integrating Large Language Models in Automated Program Verification},
  booktitle    = {The Twelfth International Conference on Learning Representations,
                  {ICLR} 2024, Vienna, Austria, May 7-11, 2024},
  publisher    = {OpenReview.net},
  year         = {2024},
  url          = {https://openreview.net/forum?id=Q3YaCghZNt},
  bibsource    = {dblp computer science bibliography, https://dblp.org}
}

@inproceedings{ASE-2024-ESBMC-AI,
  author       = {Muhammad A. A. Pirzada and
                  Giles Reger and
                  Ahmed Bhayat and
                  Lucas C. Cordeiro},
  editor       = {Vladimir Filkov and
                  Baishakhi Ray and
                  Minghui Zhou},
  title        = {LLM-Generated Invariants for Bounded Model Checking Without Loop Unrolling},
  booktitle    = {Proceedings of the 39th {IEEE/ACM} International Conference on Automated
                  Software Engineering, {ASE} 2024, Sacramento, CA, USA, October 27
                  - November 1, 2024},
  pages        = {1395--1407},
  publisher    = {{ACM}},
  year         = {2024},
  url          = {https://doi.org/10.1145/3691620.3695512},
  doi          = {10.1145/3691620.3695512},
  bibsource    = {dblp computer science bibliography, https://dblp.org}
}

@article{arXiv-2026-CILL,
  author       = {Yuheng Su and
                  Tianjun Bu and
                  Qiusong Yang and
                  Yiwei Ci and
                  Enyuan Tian},
  title        = {CIll: CTI-Guided Invariant Generation via LLMs for Model Checking},
  journal      = {CoRR},
  volume       = {abs/2602.23389},
  year         = {2026},
  url          = {https://doi.org/10.48550/arXiv.2602.23389},
  doi          = {10.48550/ARXIV.2602.23389},
  eprinttype   = {arXiv},
  eprint       = {2602.23389},
  bibsource    = {dblp computer science bibliography, https://dblp.org}
}

@inproceedings{IJCAI-2026-FindIndStrength,
  author       = {Romy Peled and
                  Daniel Kroening and
                  Michael Tautschnig and
                  Yakir Vizel},
  title        = {Large Lemma Miners: Can LLMs Do Induction Proofs for Hardware?},
  booktitle    = {Proceedings of the Thirty-Fifth International Joint Conference on
                  Artificial Intelligence and the 29th European Conference on Artificial
                  Intelligence, {IJCAI/ECAI} 2026, Bremen, Germany, 15-21 August 2026},
  pages        = {3988--3998},
  publisher    = {ijcai.org},
  year         = {2026},
  url          = {https://doi.org/10.24963/ijcai.2026/444},
  doi          = {10.24963/IJCAI.2026/444},
  bibsource    = {dblp computer science bibliography, https://dblp.org}
}

@inproceedings{FM-2026-LLM4Ind,
  author       = {Weizhi Feng and
                  Shidong Shen and
                  Jiaxiang Liu and
                  Taolue Chen and
                  Fu Song and
                  Zhilin Wu},
  editor       = {Augusto Sampaio and
                  Mari{\"{e}}lle Stoelinga},
  title        = {Can {LLM} Aid in Solving Constraints with Inductive Definitions?},
  booktitle    = {Formal Methods - 27th International Symposium, {FM} 2026, Tokyo, Japan,
                  May 18-22, 2026, Proceedings, Part {II}},
  series       = {Lecture Notes in Computer Science},
  volume       = {16557},
  pages        = {111--132},
  publisher    = {Springer},
  year         = {2026},
  url          = {https://doi.org/10.1007/978-3-032-26220-2\_6},
  doi          = {10.1007/978-3-032-26220-2\_6},
  bibsource    = {dblp computer science bibliography, https://dblp.org}
}

@article{arXiv-2026-LimICE,
  author       = {Kai Fan and
                  ShiWen Yu and
                  GuangSheng Fan and
                  HaoAng Chi and
                  WanWei Liu and
                  Ji Wang},
  title        = {LimICE: Integrating {LLM} into {ICE} Framework for Efficient Loop
                  Invariant Inference},
  journal      = {CoRR},
  volume       = {abs/2607.27606},
  year         = {2026},
  url          = {https://doi.org/10.48550/arXiv.2607.27606},
  doi          = {10.48550/ARXIV.2607.27606},
  eprinttype   = {arXiv},
  eprint       = {2607.27606},
  bibsource    = {dblp computer science bibliography, https://dblp.org}
}

@inproceedings{ASE-2024-LaM4Inv,
  author       = {Guangyuan Wu and
                  Weining Cao and
                  Yuan Yao and
                  Hengfeng Wei and
                  Taolue Chen and
                  Xiaoxing Ma},
  editor       = {Vladimir Filkov and
                  Baishakhi Ray and
                  Minghui Zhou},
  title        = {{LLM} Meets Bounded Model Checking: Neuro-symbolic Loop Invariant
                  Inference},
  booktitle    = {Proceedings of the 39th {IEEE/ACM} International Conference on Automated
                  Software Engineering, {ASE} 2024, Sacramento, CA, USA, October 27
                  - November 1, 2024},
  pages        = {406--417},
  publisher    = {{ACM}},
  year         = {2024},
  url          = {https://doi.org/10.1145/3691620.3695014},
  doi          = {10.1145/3691620.3695014},
  bibsource    = {dblp computer science bibliography, https://dblp.org}
}

@inproceedings{HOPVBenchmarks,
  author       = {Hiroyuki Katsura and
                  Naoki Kobayashi and
                  Ken Sakayori and
                  Ryosuke Sato},
  editor       = {Oleg Kiselyov},
  title        = {Mode-based Reduction from Validity Checking of Fixpoint Logic Formulas
                  to Test-Friendly Reachability Problem},
  booktitle    = {Programming Languages and Systems - 22nd Asian Symposium, {APLAS}
                  2024, Kyoto, Japan, October 22-24, 2024, Proceedings},
  series       = {Lecture Notes in Computer Science},
  volume       = {15194},
  pages        = {325--345},
  publisher    = {Springer},
  year         = {2024},
  url          = {https://doi.org/10.1007/978-981-97-8943-6\_16},
  doi          = {10.1007/978-981-97-8943-6\_16},
  bibsource    = {dblp computer science bibliography, https://dblp.org}
}

@article{HeefferHinz2017ChineseRings,
  author       = {{Heeffer, Albrecht and Hinz, Andreas M.}},
  issn         = {{2182-1976}},
  journal      = {{RECREATIONAL MATHEMATICS MAGAZINE}},
  language     = {{eng}},
  number       = {{8}},
  pages        = {{5--23}},
  title        = {{'A difficult case' : Pacioli and Cardano on the Chinese Rings}},
  url          = {{http://doi.org/10.1515/rmm-2017-0017}},
  volume       = {{4}},
  year         = {{2017}},
}

@misc{DLMFBinomialRecurrence,
         key = "{\relax DLMF}",
       title = "{\it NIST Digital Library of Mathematical Functions}",
howpublished = "\url{https://dlmf.nist.gov/}, Release 1.2.8 of 2026-09-15",
         url = "https://dlmf.nist.gov/",
        note = "F.~W.~J. Olver, A.~B. {Olde Daalhuis}, D.~W. Lozier, B.~I. Schneider,
                R.~F. Boisvert, C.~W. Clark, B.~R. Miller, B.~V. Saunders,
                H.~S. Cohl, and M.~A. McClain, eds."}

@inproceedings{Knuth1991Recursion,
  author       = {Donald E. Knuth},
  editor       = {Vladimir Lifschitz},
  title        = {Textbook Examples of Recursion},
  booktitle    = {Artificial and Mathematical Theory of Computation, Papers in Honor
                  of John McCarthy on the occasion of his sixty-fourth birthday},
  pages        = {207--229},
  publisher    = {Academic Press / Elsevier},
  year         = {1991},
  url          = {https://doi.org/10.1016/b978-0-12-450010-5.50018-9},
  doi          = {10.1016/B978-0-12-450010-5.50018-9},
  bibsource    = {dblp computer science bibliography, https://dblp.org}
}
